\documentclass[a4paper,onecolumn,accepted=2023-04-03]{quantumarticle}
\pdfoutput=1
 \usepackage[utf8]{inputenc}
\usepackage[english]{babel}
\usepackage[T1]{fontenc}
\usepackage{amsmath}
\usepackage{hyperref}
\usepackage[numbers,sort&compress]{natbib}
\usepackage{bm}
\usepackage{algorithm,algorithmic}

\usepackage{mathtools}
\usepackage{amsmath}
\usepackage[shortlabels]{enumitem}

\usepackage{graphicx,epic,eepic,epsfig,amsmath,latexsym,amssymb,verbatim,color}
 
\usepackage{amsfonts}       
\usepackage{nicefrac}       

\usepackage{amsmath}
\usepackage{bbm}
\usepackage{float}
\usepackage{tikz}
\usetikzlibrary{chains}
\usetikzlibrary{fit}
\usepackage{pgflibraryarrows}		
\usepackage{pgflibrarysnakes}		

\usepackage{epsfig}
\usetikzlibrary{shapes.symbols,patterns} 
\usepackage{pgfplots}

\usepackage{theorem}

\newtheorem{proposition}{Proposition}
\newtheorem{lemma}[proposition]{Lemma}

\newtheorem{theorem}[proposition]{Theorem}

\def\squareforqed{\hbox{\rlap{$\sqcap$}$\sqcup$}}
\def\qed{\ifmmode\squareforqed\else{\unskip\nobreak\hfil
\penalty50\hskip1em\null\nobreak\hfil\squareforqed
\parfillskip=0pt\finalhyphendemerits=0\endgraf}\fi}
\def\endenv{\ifmmode\;\else{\unskip\nobreak\hfil
\penalty50\hskip1em\null\nobreak\hfil\;
\parfillskip=0pt\finalhyphendemerits=0\endgraf}\fi}
\newenvironment{proof}{\noindent \textbf{{Proof~} }}{\hfill $\blacksquare$}

\newcounter{remark}
\newenvironment{remark}[1][]{\refstepcounter{remark}\par\medskip\noindent%
\textbf{Remark~\theremark #1} }{\medskip}

\newcounter{example}

\mathchardef\ordinarycolon\mathcode`\:
\mathcode`\:=\string"8000
\def\vcentcolon{\mathrel{\mathop\ordinarycolon}}
\begingroup \catcode`\:=\active
  \lowercase{\endgroup
  \let :\vcentcolon
  }

\usepackage{cleveref}
\usepackage{graphicx}
\usepackage{xcolor}

\RequirePackage[framemethod=default]{mdframed}
\newmdenv[skipabove=7pt,
skipbelow=7pt,
backgroundcolor=darkblue!15,
innerleftmargin=5pt,
innerrightmargin=5pt,
innertopmargin=5pt,
leftmargin=0cm,
rightmargin=0cm,
innerbottommargin=5pt,
linewidth=1pt]{tBox}

\newmdenv[skipabove=7pt,
skipbelow=7pt,
backgroundcolor=blue2!25,
innerleftmargin=5pt,
innerrightmargin=5pt,
innertopmargin=5pt,
leftmargin=0cm,
rightmargin=0cm,
innerbottommargin=5pt,
linewidth=1pt]{dBox}
\newmdenv[skipabove=7pt,
skipbelow=7pt,
backgroundcolor=darkkblue!15,
innerleftmargin=5pt,
innerrightmargin=5pt,
innertopmargin=5pt,
leftmargin=0cm,
rightmargin=0cm,
innerbottommargin=5pt,
linewidth=1pt]{sBox}
\definecolor{darkblue}{RGB}{0,76,156}
\definecolor{darkkblue}{RGB}{0,0,153}
\definecolor{blue2}{RGB}{102,178,255}
\definecolor{darkred}{RGB}{195,0,0}

\newcommand{\nc}{\newcommand}
\nc{\rnc}{\renewcommand}
\nc{\beg}{\begin{equation}}
\nc{\eeq}{{\end{equation}}}
\nc{\beqa}{\begin{eqnarray}}
\nc{\eeqa}{\end{eqnarray}}
\nc{\lbar}[1]{\overline{#1}}
\nc{\bra}[1]{\langle#1|}
\nc{\ket}[1]{|#1\rangle}
\nc{\ketbra}[2]{|#1\rangle\!\langle#2|}
\nc{\braket}[2]{\langle#1|#2\rangle}

\nc{\proj}[1]{| #1\rangle\!\langle #1 |}
\nc{\avg}[1]{\langle#1\rangle}
\nc{\rank}{\operatorname{Rank}}
\nc{\smfrac}[2]{\mbox{$\frac{#1}{#2}$}}
\nc{\tr}{\operatorname{Tr}}
\nc{\ox}{\otimes}
\nc{\dg}{\dagger}
\nc{\dn}{\downarrow}
\nc{\cA}{{\cal A}}
\nc{\cB}{{\cal B}}
\nc{\cC}{{\cal C}}
\nc{\cD}{{\cal D}}
\nc{\cE}{{\cal E}}
\nc{\cF}{{\cal F}}
\nc{\cG}{{\cal G}}
\nc{\cH}{{\cal H}}
\nc{\cI}{{\cal I}}
\nc{\cJ}{{\cal J}}
\nc{\cK}{{\cal K}}
\nc{\cL}{{\cal L}}
\nc{\cM}{{\cal M}}
\nc{\cN}{{\cal N}}
\nc{\cO}{{\cal O}}
\nc{\cP}{{\cal P}}
\nc{\cQ}{{\cal Q}}
\nc{\cR}{{\cal R}}
\nc{\cS}{{\cal S}}
\nc{\cT}{{\cal T}}
\nc{\cU}{{\cal U}}
\nc{\cV}{{\cal V}}
\nc{\cX}{{\cal X}}
\nc{\cY}{{\cal Y}}
\nc{\cZ}{{\cal Z}}
\nc{\cW}{{\cal W}}
\nc{\csupp}{{\operatorname{csupp}}}
\nc{\qsupp}{{\operatorname{qsupp}}}
\nc{\var}{{\operatorname{var}}}
\nc{\rar}{\rightarrow}
\nc{\lrar}{\longrightarrow}
\nc{\polylog}{{\operatorname{polylog}}}
\nc{\wt}{{\operatorname{wt}}}
\nc{\av}[1]{{\left\langle {#1} \right\rangle}}
\nc{\supp}{{\operatorname{supp}}}

\nc{\argmin}{{\operatorname{argmin}}}

\def\x{\xi}

\nc{\RR}{{{\mathbb R}}}
\nc{\CC}{{{\mathbb C}}}
\nc{\FF}{{{\mathbb F}}}
\nc{\NN}{{{\mathbb N}}}
\nc{\ZZ}{{{\mathbb Z}}}
\nc{\PP}{{{\mathbb P}}}
\nc{\QQ}{{{\mathbb Q}}}
\nc{\UU}{{{\mathbb U}}}
\nc{\EE}{{{\mathbb E}}}
\nc{\id}{{\operatorname{id}}}

\nc{\CHSH}{{\operatorname{CHSH}}}

\nc{\be}{\begin{equation}}
\nc{\ee}{{\end{equation}}}
\nc{\bea}{\begin{eqnarray}}
\nc{\eea}{\end{eqnarray}}
\nc{\<}{\langle}
\rnc{\>}{\rangle}
\nc{\rU}{\mbox{U}}

\nc{\ob}[1]{#1}

\nc{\SEP}{{\text{\rm SEP}}}
\nc{\NS}{{\text{\rm NS}}}
\nc{\LOCC}{{\text{\rm LOCC}}}
\nc{\PPT}{{\text{\rm PPT}}}
\nc{\EXT}{{\text{\rm EXT}}}
\nc{\Sym}{{\operatorname{Sym}}}

\nc{\ERLO}{{E_{\text{r,LO}}}}
\nc{\ERLOCC}{{E_{\text{r,LOCC}}}}
\nc{\ERPPT}{{E_{\text{r,PPT}}}}
\nc{\ERLOCCinfty}{{E^{\infty}_{\text{r,LOCC}}}}
\nc{\Aram}{{\operatorname{\sf A}}}

\newcommand{\Choi}{Choi-Jamio\l{}kowski }

\usepackage{tikz}

\makeatletter
\def\grd@save@target#1{%
  \def\grd@target{#1}}
\def\grd@save@start#1{%
  \def\grd@start{#1}}
\tikzset{
  grid with coordinates/.style={
    to path={%
      \pgfextra{%
        \edef\grd@@target{(\tikztotarget)}%
        \tikz@scan@one@point\grd@save@target\grd@@target\relax
        \edef\grd@@start{(\tikztostart)}%
        \tikz@scan@one@point\grd@save@start\grd@@start\relax
        \draw[minor help lines,magenta] (\tikztostart) grid (\tikztotarget);
        \draw[major help lines] (\tikztostart) grid (\tikztotarget);
        \grd@start
        \pgfmathsetmacro{\grd@xa}{\the\pgf@x/1cm}
        \pgfmathsetmacro{\grd@ya}{\the\pgf@y/1cm}
        \grd@target
        \pgfmathsetmacro{\grd@xb}{\the\pgf@x/1cm}
        \pgfmathsetmacro{\grd@yb}{\the\pgf@y/1cm}
        \pgfmathsetmacro{\grd@xc}{\grd@xa + \pgfkeysvalueof{/tikz/grid with coordinates/major step}}
        \pgfmathsetmacro{\grd@yc}{\grd@ya + \pgfkeysvalueof{/tikz/grid with coordinates/major step}}
        \foreach \x in {\grd@xa,\grd@xc,...,\grd@xb}
        \node[anchor=north] at (\x,\grd@ya) {\pgfmathprintnumber{\x}};
        \foreach \y in {\grd@ya,\grd@yc,...,\grd@yb}
        \node[anchor=east] at (\grd@xa,\y) {\pgfmathprintnumber{\y}};
      }
    }
  },
  minor help lines/.style={
    help lines,
    step=\pgfkeysvalueof{/tikz/grid with coordinates/minor step}
  },
  major help lines/.style={
    help lines,
    line width=\pgfkeysvalueof{/tikz/grid with coordinates/major line width},
    step=\pgfkeysvalueof{/tikz/grid with coordinates/major step}
  },
  grid with coordinates/.cd,
  minor step/.initial=.2,
  major step/.initial=1,
  major line width/.initial=2pt,
}
\makeatother

\usepackage{thmtools}
\usepackage{thm-restate}
\usepackage{etoolbox}
\makeatletter
\def\problem@s{}
\newcounter{problems@cnt}

\newcommand{\allproblems}{\problem@s}
\makeatother

\definecolor{beamer}{rgb}{0.2,0.2,0.7}
\definecolor{colorone}{rgb}{1,0.36,0.03}
\definecolor{colortwo}{rgb}{0.4,0.77,0.17}
\definecolor{colorthree}{rgb}{0.01,0.51,0.93}
\definecolor{colorfour}{rgb}{0.47,0.26,0.58}
\definecolor{colorfive}{rgb}{0.12,0.55,0.16}
\usepackage{tcolorbox}
\usepackage{relsize}
\usepackage{graphicx}
\usepackage{booktabs}
\usepackage{amsmath}
\usepackage{tikz}
\usetikzlibrary{quantikz}
\nc{\st}{\text{subject to} \ }
\nc{\supre}{\text{supremum} \ }
\nc{\sdp}{\text{sdp}}
\usepackage{array}

\newcommand{\sgn}{{\rm sgn}}
\allowdisplaybreaks

\nc{\ith}[1]{{#1}^\mathrm{th}}

\newcommand{\update}[1]{\textcolor{black}{#1}}

\begin{document}
\title{Information recoverability of noisy quantum states}
 
\author{Xuanqiang Zhao}
\affiliation{Institute for Quantum Computing, Baidu Research, Beijing 100193, China}
\affiliation{QICI Quantum Information and Computation Initiative, Department of Computer Science, The University of Hong Kong, Pokfulam Road, Hong Kong, China}

\author{Benchi Zhao}
\affiliation{Institute for Quantum Computing, Baidu Research, Beijing 100193, China}

\author{Zihan Xia}
\affiliation{Institute for Quantum Computing, Baidu Research, Beijing 100193, China}

\author{Xin Wang}
\email{wangxin73@baidu.com}
\homepage{https://www.xinwang.info/}
\affiliation{Institute for Quantum Computing, Baidu Research, Beijing 100193, China}

\begin{abstract}
 Extracting classical information from quantum systems is an essential step of many quantum algorithms. However, this information could be corrupted as the systems are prone to quantum noises, and its distortion under quantum dynamics has not been adequately investigated. In this work, we introduce a systematic framework to study how well we can retrieve information from noisy quantum states. Given a noisy quantum channel, we fully characterize the range of recoverable classical information. This condition allows a natural measure quantifying the information recoverability of a channel. Moreover, we resolve the minimum information retrieving cost, which, along with the corresponding optimal protocol, is efficiently computable by semidefinite programming. As applications, we establish the limits on the information retrieving cost for practical quantum noises and employ the corresponding protocols to mitigate errors in ground state energy estimation.
Our work gives the first full characterization of information recoverability of noisy quantum states from the recoverable range to the recovering cost, revealing the ultimate limit of probabilistic error cancellation.
\end{abstract}

\maketitle

\section{Introduction}
\subsection{Background}
The dynamics of a closed quantum system is described as a unitary evolution~\cite{Griffiths2018}. However, quantum systems are \update{rarely} closed in practice as they unavoidably interacts with the environment. Quantum channels, stemming from the unitary dynamics in a larger Hilbert space, are considered as the proper mathematical formalism depicting the evolution of general quantum systems~\cite{Nielsen2010}. Quantum channels represent the quantum information manipulation processes and are essential to quantum computation~\cite{Nielsen2010,Wilde2017book,Watrous2011b,Bennett2014,Wang2019,Smith2008,WWS19}.

A central \update{subroutine} of quantum computation is to extract classical information from a quantum system. The expectation value of some chosen observable, also known as shadow information, \update{characterizes} physical properties of the quantum system, making estimating expectation values the main goal of many quantum algorithms~\cite{Nielsen2010,Cerezo2020}, such as variational quantum eigensolver (VQE)~\cite{Peruzzo2014}. The importance of expectation values has motivated the study of shadow tomography and related problems~\cite{aaronson2019shadow,aaronson2019gentle,huang2020predicting}.
However, inevitable noises modeled as quantum channels can corrupt the shadow information as they undesirably change the state of the quantum system, preventing us from estimating the expectation value accurately.

It is natural to ask how a quantum channel $\cN$ affects the information stored in quantum states. The minimum fidelity~\cite{Barnum2000,Bennett1997,DiVincenzo1998} of some initial state $\proj{\psi}$ and the final state $\cN(\proj{\psi})$ is the one we often use to answer this question.
But this method is not appropriate for quantifying how well a quantum channel preserves shadow information. Take the completely phase damping channel $\cN_{\rm PD}(\rho) = Z\rho Z$ as an example. While $F_{\min}(\cN_{\rm PD}) = 0$, the expectation value of a diagonal observable (e.g., observable $Z$) is unaffected. Similarly, other fidelity-based measures such as average fidelity and entanglement fidelity~\cite{schumacher1996sending,Barnum2000} are not suitable to depict the shadow information preservation either. Some other measures for information preservation are various channel capacities (e.g., classical capacity~\cite{holevo1973bounds, schumacher1997sending, holevo1998capacity}), where each channel capacity quantifies a quantum channel's utility to transfer information for a certain purpose. 
However, these capacities are derived in asymptotic settings with multiple uses of the channel, and thus are not proper guides to practical tasks with finite quantum resources to some extent. Naturally, the following key questions then arise:
\begin{enumerate}
    \item \textit{What is a proper way to quantify the level of shadow information preservation by a\\
    quantum channel?}
    \item \textit{How to quantify the cost of retrieving a certain piece of shadow information given\\
    it is preserved?}
\end{enumerate}

\subsection{Contributions}
To address these two questions in an operational way, we introduce a framework for retrieving shadow information from a noisy state.
Both practically relevant and theoretically interesting, this framework combines an efficient way to manipulate quantum states and general quantum operations. Note that similar protocols are employed in predicting properties of quantum states~\cite{Buscemi2013,huang2020predicting} and mitigating quantum errors~\cite{Temme2017,Endo2018a,Takagi2020,Jiang2021physical,Piveteau2021}. In particular, we suppose that multiple copies of a noisy state $\cN(\rho)$ are available. For each copy we apply a quantum channel randomly sampled from an ensemble of channels $\{\cD_j\}$ and then measure the observable of interest, as shown in Fig. \ref{fig:setting}. We say that a channel $\cN$ preserves the classical information inquired by an observable $O$ if we can recover this information, that is: there exist a set of quantum channels $\{\cD_j\}$ and a set of real numbers $\{c_j\}$ such that
\begin{align}\label{eq:linear_comb_of_qchannel}
    \sum_j c_j\tr[\cD_j\circ\cN(\rho) O] = \tr[\rho O].
\end{align}
We establish a necessary and sufficient condition for the preservation of the desired shadow information in terms of a relation between $O$ and $\cN$ that should be satisfied. This condition motivates a measure called \textit{shadow destructivity} characterizing a quantum channel's level of destroying shadow information.

We define a measure called \textit{retrieving cost} to quantify the minimum cost of recovering $\tr[\rho O]$, which also quantifies how well the channel preserves the shadow information queried by a specific observable. This measure, along with the concrete retrieving protocol, is efficiently computable with respect to the system dimension via semidefinite programming~\cite{Vandenberghe1996semidefinite}. We analytically obtain the values of this measure and the retrieving protocols for generalized amplitude damping (GAD) channels and depolarizing channels. Our retrieving costs set ultimate limits on the cost required for shadow retrieving, and the corresponding protocols outperform existing probabilistic error cancellation (PEC) methods~\cite{Temme2017,Endo2018a,Takagi2020,Jiang2021physical,Piveteau2021}.
Specifically, we employ our method to estimate the ground state energies of several molecules with VQE. The results show that the gap of sampling overhead between our method and a conventional PEC method gets larger as the size of the quantum system grows, implying potential applications of our method in implementing near-term algorithms.

\begin{figure}[t]
    \centering
    \includegraphics[width=0.5\linewidth]{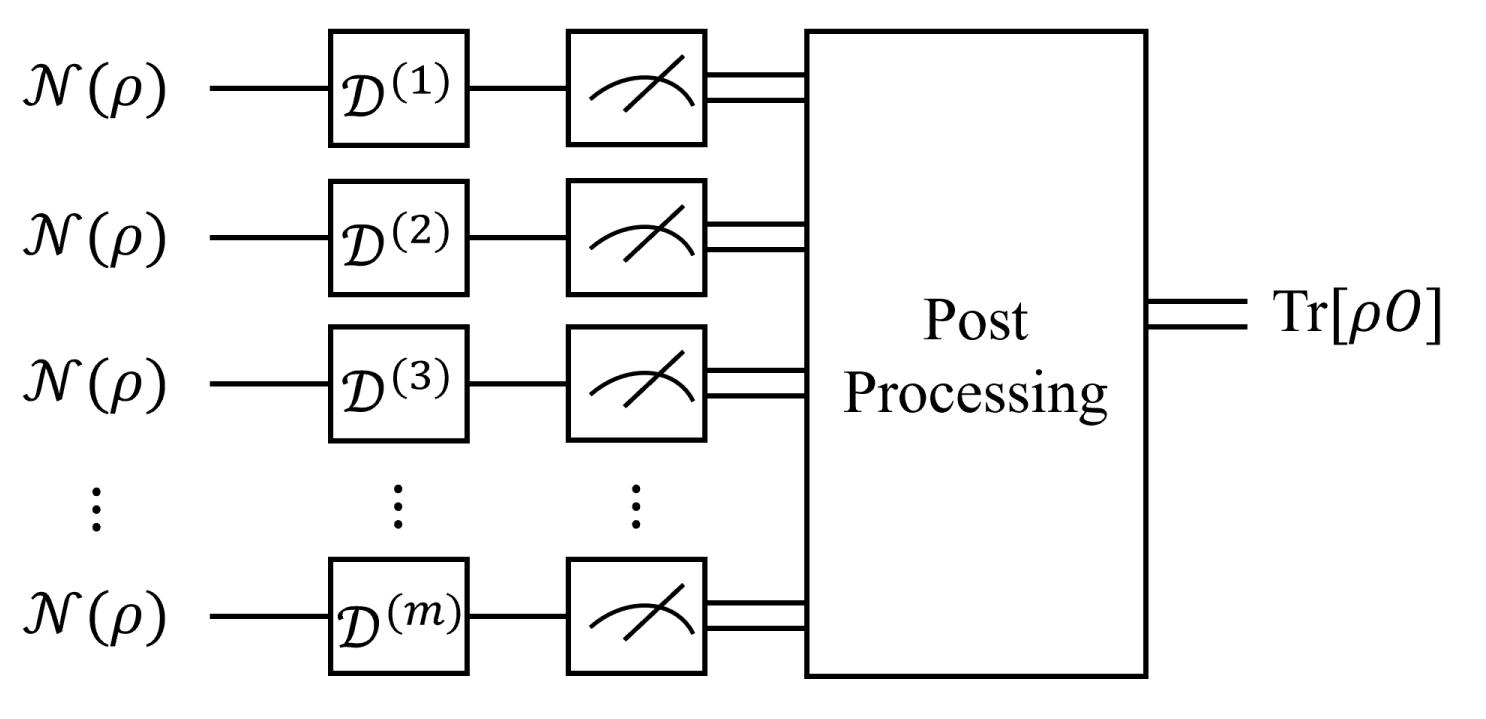}
    \caption{Framework for retrieving shadow information from noisy quantum states. After independently applying $\cD^{(i)}$ randomly selected from $\{\cD_j\}$ to multiple copies of $\cN(\rho)$, estimation of the desired information $\tr[\rho O]$ is obtained via making measurements with the corresponding observable $O$ and post-processing.}
    \label{fig:setting} 
\end{figure}

\subsection{Preliminaries}
Before introducing our results, we set the notations and define several quantities that will be used throughout this paper.
We use symbols such as $\mathcal{H}_A$ and $\mathcal{H}_B$ to denote finite-dimensional Hilbert spaces associated with systems $A$ and $B$, respectively.
We use $d_A$ to denote the dimension of the system $A$. The sets of linear and Hermitian operators acting on $A$ are denoted by $\cL_A$ and $\cL^{\rm H}_A$, respectively.

In this work, we focus on linear maps having the same input and output dimension. A linear map $\cN_{A\to A'}$ transforms linear operators in $\cL_A$ to linear operators in $\cL_{A'}$, where $\cH_{A'}$ is isomorphic to $\cH_A$.
We call a linear map $\cN_{A\to A'}$ a quantum channel if it is completely positive and trace-preserving (CPTP). By saying $\cN$ is completely positive (CP), we mean $\id_R\ox\cN$ is a positive map with a reference system $R$ of arbitrary dimension. By saying $\cN$ is trace-preserving (TP), we mean $\tr[\cN(X)] = \tr[X]$ for all $X\in\cL_A$.
A linear map $\cN$ is called Hermitian-preserving (HP) if $\cN(X)\in\cL^{\rm H}_{A'}$ for all $X\in\cL^{\rm H}_A$. For any linear map $\cN_{A\to A'}$, its \Choi matrix is given by $J_{\cN}\equiv \sum_{i,j=0}^{d_A-1} \ket{i}\bra{j} \ox \cN_{A\to A'} (\ket{i}\bra{j})$, where $\{\ket{i}\}_{i=0}^{d_A-1}$ is an orthonormal basis in $\cH_A$.

\section{Main Results}
\subsection{A Necessary and Sufficient Condition}
We study the preservation of shadow information by considering an operational task: recovering the expectation value $\tr[\rho O]$ of an observable $O$ from multiple copies of a corrupted state $\cN(\rho)$, where $\cN$ is the noisy channel of interest and $\rho$ is an unknown state. A usual way~\cite{Jiang2021physical} is to simulate a Hermitian-preserving and trace-preserving (HPTP) map $\cD$, which we call a retriever in this work, such that $\cD\circ\cN = \id$, where $\id$ is the identity map. While this approach is sufficient for recovering the target information, its requirement on the retriever $\cD$ is generally not necessary as we are only concerned with the expectation value of a specific observable. Hence, the only necessary requirement for a retriever $\cD$ is
\begin{align}\label{eq:necessary_req}
    \tr [\cD\circ\cN(\rho) O] = \tr[\rho O]
\end{align}
for any state $\rho$ and a fixed observable $O$. Our framework extends the set of potential retrievers from HPTP maps to Hermitian-preserving and trace-scaling (HPTS) maps for that HPTS maps are physically simulatable in the way prescribed by Eq.~\eqref{eq:linear_comb_of_qchannel}, as shown in Lemma~\ref{lemma:hpts}. We say a linear map $\cN$ is trace-scaling (TS) if $\tr_{A'}[J_\cN] = pI_A$ for some real number $p$, and HPTP maps are special cases of HPTS maps when $p=1$.

\begin{figure*}[t]
    \centering
    \includegraphics[width=0.8\textwidth]{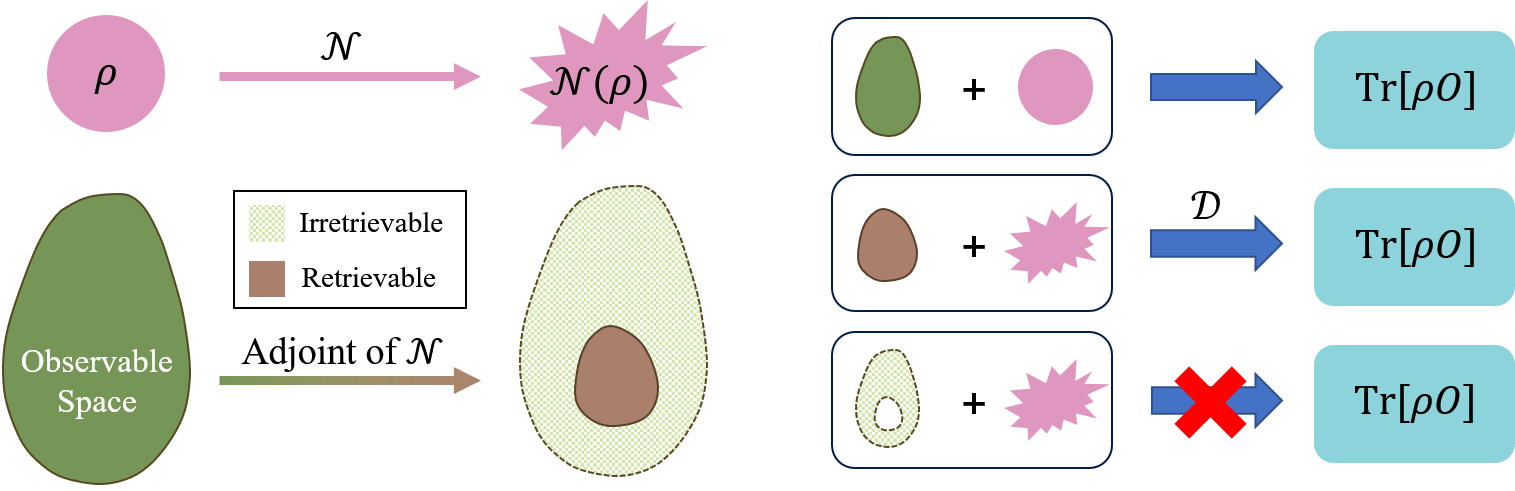}
    \caption{Illustration of Theorem~\ref{thm:applicability}. Suppose a state $\rho$ is corrupted by a channel $\cN$. A piece of shadow information is recoverable by the retriever $\cD$ if and only if it is queried by an observable lying in the brown area, which represents the observable space evolved under the adjoint of $\cN$. This result connects channel corrupting state in Schr\"odinger picture with the evolving of the observable space in Heisenberg picture.}
    \label{fig:main_result}
\end{figure*}

Here, given a quantum channel $\cN$ and an observable $O$, we give the only condition that a retriever $\cD$ needs to satisfy so that Eq.~\eqref{eq:necessary_req} holds for an arbitrary state:
\begin{align}\label{eq:min_req_on_retriever}
    \cN^\dagger\circ\cD^\dagger(O) = O,
\end{align}
where $\cN^\dagger$ and $\cD^\dagger$ are the adjoint maps of $\cN$ and $\cD$, respectively.
This condition is a direct implication of Lemma~\ref{lemma:minimum_requirement}.

\begin{lemma}\label{lemma:minimum_requirement}
Given an observable $O$, an HP map $\cM$ satisfies $\tr[\cM(\rho) O] = \tr[\rho O]$ for any state $\rho$ if and only if it holds that $\cM^\dagger(O) = O$.
\end{lemma}
\begin{proof}
First, we assume that $\cM^\dagger(O) = O$ is true. For $\cM$ being HP, its adjoint map $\cM^\dagger$ is also HP (see Remark~\ref{remark:hp_adjoint}). Hence, $\tr[\cM(\rho) O] = \tr[\rho (\cM^\dagger(O^\dagger))^\dagger] = \tr[\rho \cM^\dagger(O)]$ for both $O$ and $\cM^\dagger(O)$ being Hermitian. Then, with our assumption, we have $\tr[\cM(\rho) O] = \tr[\rho \cM^\dagger(O)] = \tr[\rho O]$.

For the other direction, given $\tr[\cM(\rho) O] = \tr[\rho O]$ for all states $\rho$, we immediately have $\tr[\rho \cM^\dagger(O)] = \tr[\rho O]$. Suppose, for the sake of contradiction, that $\cM^\dagger(O)\neq O$. If we define $\Delta O = \cM^\dagger(O) - O$, then $\tr[\rho \Delta O] = 0$ for all $\rho$. However, this implies that $\Delta O = 0$ since otherwise we can choose a $\sigma$ from the nonzero eigenspace of $\Delta O$ so that $\tr[\sigma \Delta O] \neq 0$. As $\Delta O = 0$ contradicts with the assumption that $\cM^\dagger(O)\neq O$, we conclude that $\tr[\rho \cM^\dagger(O)] = \tr[\rho O]$ holds for all $\rho$ only if $\cM^\dagger(O) = O$.
\end{proof}

\begin{remark}\label{remark:hp_adjoint}
A linear map $\cN$ is Hermitian-preserving (HP) if and only if $J_\cN$ is Hermitian. The adjoint of an HP map $\cN$ is also HP. This can be seen by checking the Hermiticity of its Choi matrix: $J_{\cN^\dagger}^\dagger = (J_\cN^T)^\dagger = J_\cN^T = J_{\cN^\dagger}$ for $J_\cN$ being Hermitian.
\end{remark}

Lemma~\ref{lemma:minimum_requirement} sets the criterion that we should refer to when searching for a retriever. We say a channel $\cN$ preserves the shadow information queried by an observable $O$ if there exists a retriever $\cD$ satisfying Eq.~\eqref{eq:min_req_on_retriever}. It is crucial to ask for what observables the channel preserves their corresponding shadow information.
In Theorem~\ref{thm:applicability}, we establish a necessary and sufficient condition for the preservation of the shadow information by introducing the adjoint image of $\cL^{\rm H}_{A'}$ under the channel $\cN_{A\to A'}$, which we define to be the image of $\cL^{\rm H}_{A'}$ under $\cN^\dagger_{A'\to A}$ and is denoted by $\cN^\dagger(\cL^{\rm H}_{A'})$.

\begin{theorem}{\rm (Necessary and sufficient condition for information preservation).}\label{thm:applicability}
Given a quantum channel $\cN_{A\to A'}$ and an observable $O$, there exists an HPTS map $\cD$ such that $\cM \equiv \cD\circ\cN$ satisfies $\cM^\dagger(O) = O$ if and only if $O \in \cN^\dagger(\cL^{\rm H}_{A'})$.
\end{theorem}

The proof for the necessity part is rather straightforward, while that for the sufficiency part is nontrivial. To prove the ``if'' part, we denote by $Q$ a Hermitian operator that evolves into the observable $O$ under the map $\cN^\dagger$, i.e., $\cN^\dagger(Q) = O$. Then, it is sufficient to construct an HPTS map $\cD$ such that $\cD^\dagger(O) = Q$. According to Lemma~\ref{lemma:hpts_adjoint}, we can instead construct an HP and unit-scaling map $\cD^\dagger$, which guarantees that the map $\cD$ is HPTS. By saying that $\cD^\dagger$ is unit-scaling, we mean that $\cD^{\dagger}$ scales the identity operator. The main difficulty of this construction is to ensure that $\cD^\dagger$ is unit-scaling. To overcome this difficulty, we split all the observables into three categories by their rank and trace.

We first consider observables that do not have full rank and are not traceless, i.e., $\tr[O] \neq 0$. For these observables, we exploit their kernel so that the constructed $\cD^\dagger$ is unit-scaling. With some manipulation, the other two cases can be reduced to this one. The detailed proof is given below.

\begin{proof}
For the ``only if'' part, suppose such a map $\cD$ exists for the given channel $\cN$ and non-zero observable $O$. Note that $\cD^\dagger$ is also a Hermitian-preserving map since its Choi operator is Hermitian, as we have noted in Remark~\ref{remark:hp_adjoint}. Then, we can let $Q = \cD^\dagger(O)$ so that $\cN^\dagger(Q) = \cN^\dagger \circ \cD^\dagger(O) = \cM^\dagger(O) = O$. Hence, $O \in \cN^\dagger(\cL^{\rm H}_{A'})$ if there exists an HPTS map $\cD$ such that $\cM \equiv \cD\circ\cN$ satisfies $\cM^\dagger(O) = O$.

For the ``if'' part, suppose $O \in \cN^\dagger(\cL^{\rm H}_{A'})$, which means that there exists a Hermitian operator $Q$ such that $\cN^\dagger(Q) = O$. It is sufficient to construct an HPTS map $\cD$ satisfying $\cD^\dagger(O) = Q$ so that $\cN^\dagger(Q) = \cN^\dagger \circ \cD^\dagger(O) = \cM^\dagger(O) = O$. By Lemma~\ref{lemma:hpts_adjoint}, it is equivalent to construct a Hermitian-preserving and unit-scaling map $\cD^\dagger$ satisfying the above requirement. Denoting the rank of the observable $O$ by $k$, the trace of $O$ by $t$, and the dimension of the Hilbert space $\cH_A$ by $d$, we complete the proof by constructing a Hermitian-preserving and unit-scaling map $\cD^\dagger$ satisfying $\cD^\dagger(O) = Q$ in each of the following three cases.
\begin{itemize}
    \item Case 1: $k<d$ and $t\neq0$.
    Let $P$ be the projection on the support of $O$ and $P^\perp$ be the projection on the kernel of $O$. We can construct a map $\cD^\dagger$ such that
    \begin{align}
        \cD^\dagger(\cdot) = \frac{Q}{t}\tr[P\cdot P] + \frac{k(I-Q)}{t(d-k)}\tr[P^\perp\cdot P^\perp].
    \end{align}
    Clearly, $\cD^\dagger$ is an HP map. Note that
    \begin{align}
        \cD^\dagger(O) &= \frac{Q}{t}\tr[POP] + \frac{k(I-Q)}{t(d-k)}\tr[P^\perp OP^\perp]\\
        &= \frac{Q}{t}\tr[O]\\
        &= Q
    \end{align}
    and
    \begin{align}
        \cD^\dagger(I) &= \frac{Q}{t}\tr[PIP] + \frac{k(I-Q)}{t(d-k)}\tr[P^\perp IP^\perp]\\
        &= \frac{Q}{t}\tr[P] + \frac{k(I-Q)}{t(d-k)}\tr[P^\perp]\\
        &= \frac{Q}{t}\cdot k + \frac{k(I-Q)}{t(d-k)}\cdot(d-k)\\
        &= \frac{k}{t}I,
    \end{align}
    which implies that $\cD^\dagger$ is also unit-scaling and $\cD^\dagger(O) = Q$.

    \item Case 2: $k<d$ and $t=0$.
    Since $t=0$, $O$ has at least two distinct eigenvalues. Let $O = \sum_j \lambda_j \proj{\psi_j}$ be the spectral decomposition of $O$, where each $\lambda_j\neq 0$. Define $\Tilde{O}\equiv O-\lambda_0\proj{\psi_0}$ and $\Tilde{Q}\equiv Q+\Delta Q$, where $\Delta Q$ is a Hermitian operator that we will fix later. Let \update{$\Tilde{k}$} denote the rank of $\Tilde{O}$ and $\Tilde{t}$ the trace of $\Tilde{O}$. Following the definition of $\Tilde{O}$, we have $\Tilde{k} = k-1$ and $\Tilde{t} = -\lambda_0$. As $\Tilde{k}<d$ and $\Tilde{t}\neq0$, we can construct a Hermitian-preserving and unit-scaling map
    \begin{align}
        \cD^\dagger(\cdot) &= \frac{\Tilde{Q}}{\Tilde{t}}\tr[\Tilde{P}\cdot \Tilde{P}] + \frac{\Tilde{k}(I-\Tilde{Q})}{\Tilde{t}(d-\Tilde{k})}\tr[\Tilde{P}^\perp\cdot \Tilde{P}^\perp]\\
        &= -\frac{\Tilde{Q}}{\lambda_0}\tr[\Tilde{P}\cdot \Tilde{P}] - \frac{(k-1)(I-\Tilde{Q})}{\lambda_0(d-k+1)}\tr[\Tilde{P}^\perp\cdot \Tilde{P}^\perp],
    \end{align}
    where $\Tilde{P}$ denotes the projection on the support of $\Tilde{O}$ and $\Tilde{P}^\perp$ be the projection on the kernel of $\Tilde{O}$. By the definition of $\Tilde{O}$, we have $\Tilde{P}O\Tilde{P} = \sum_{j>0} \lambda_j\proj{\psi_j}$ and $\Tilde{P}^\perp O\Tilde{P}^\perp = \lambda_0\proj{\psi_0}$, and thus $\tr[\Tilde{P}O\Tilde{P}] = -\lambda_0$ and $\tr[\Tilde{P}^\perp O\Tilde{P}^\perp] = \lambda_0$. Hence,
    \begin{align}
        \cD^\dagger(O) &= -\frac{\Tilde{Q}}{\lambda_0}\cdot(-\lambda_0) - \frac{(k-1)(I-\Tilde{Q})}{\lambda_0(d-k+1)}\cdot\lambda_0\\
        &= \frac{d}{d-k+1}Q + \frac{d}{d-k+1}\Delta Q - \frac{k-1}{d-k+1}I.
    \end{align}
    Setting $\Delta Q = \frac{k-1}{d}(I-Q)$, we will have $\cD^\dagger(O) = Q$.

    \item Case 3: $k=d$.
    We first consider the case where $O = cI$ for some real coefficient $c$. As $Q$ can be any Hermitian operator satisfying $\cN^\dagger(Q) = O$, and we know that $\cN^\dagger$ is unital, we can let $Q = O$ so that $\cD^\dagger$ being the identity map $\id$ will complete the proof.
    
    Now, consider the case where $O\neq cI$ for any real coefficient $c$, which implies that $O$ must have at least two distinct eigenvalues. Let $O = \sum_j \lambda_j \proj{\psi_j}$ be the spectral decomposition of $O$, where each $\lambda_j\neq 0$ and $\lambda_0$ is the smallest eigenvalue. Then, we define $\Tilde{O}\equiv O-\lambda_0I$ and $\Tilde{Q}\equiv Q+\Delta Q$, where $\Delta Q$ is a Hermitian operator that we will fix later. Let \update{$\Tilde{k}$} denote the rank of $\Tilde{O}$ and $\Tilde{t}$ the trace of $\Tilde{O}$. Following the definition of $\Tilde{O}$, we have $\Tilde{k} < d$ and $\Tilde{t} = t-d\cdot\lambda_0 > 0$. Hence, by Case 1, we can construct a Hermitian-preserving and unit-scaling map
    \begin{align}
        \cD^\dagger(\cdot) &= \frac{\Tilde{Q}}{\Tilde{t}}\tr[\Tilde{P}\cdot \Tilde{P}] + \frac{\Tilde{k}(I-\Tilde{Q})}{\Tilde{t}(d-\Tilde{k})}\tr[\Tilde{P}^\perp\cdot \Tilde{P}^\perp]\\
        &= \frac{\Tilde{Q}}{t-d\cdot\lambda_0}\tr[\Tilde{P}\cdot \Tilde{P}] + \frac{\Tilde{k}(I-\Tilde{Q})}{(t-d\cdot\lambda_0)(d-\Tilde{k})}\tr[\Tilde{P}^\perp\cdot \Tilde{P}^\perp],
    \end{align}
    where $\Tilde{P}$ denotes the projection on the support of $\Tilde{O}$ and $\Tilde{P}^\perp$ be the projection on the kernel of $\Tilde{O}$. By the definition of $\Tilde{O}$, we have $\Tilde{P}O\Tilde{P} = \sum_{\lambda_j\neq\lambda_0} \lambda_j\proj{\psi_j}$ and $\Tilde{P}^\perp O\Tilde{P}^\perp = \sum_{\lambda_j=\lambda_0} \lambda_j\proj{\psi_j}$, and thus $\tr[\Tilde{P}O\Tilde{P}] = t-(d-\Tilde{k})\lambda_0$ and $\tr[\Tilde{P}^\perp O\Tilde{P}^\perp] = (d-\Tilde{k})\lambda_0$. Hence,
    \begin{align}
        \cD^\dagger(O) &= \frac{\Tilde{Q}}{t-d\cdot\lambda_0}\cdot(t-(d-\Tilde{k})\lambda_0) + \frac{\Tilde{k}(I-\Tilde{Q})}{(t-d\cdot\lambda_0)(d-\Tilde{k})}\cdot(d-\Tilde{k})\lambda_0\\
        &= Q + \Delta Q + \frac{\Tilde{k}\cdot\lambda_0}{t-d\cdot\lambda_0}I.
    \end{align}
    Setting $\Delta Q = -\frac{\Tilde{k}\cdot\lambda_0}{t-d\cdot\lambda_0}I$ results in $\cD^\dagger(O) = Q$.
\end{itemize}
Since any observable $O$ can be categorized into one of the three cases above, we conclude that given a quantum channel $\cN$, if $O\in{\rm Image(\cN^\dagger)}$, then there exists a Hermitian-preserving and unit-scaling map $\cD^\dagger$ such that $\cD^\dagger(O) = Q$ for some $Q$ satisfying $\cN^\dagger(Q) = O$. As the adjoint of a Hermitian-preserving and unit-scaling map is HPTS, we can always find an HPTS retriever $\cD$ so that $\cN^\dagger\circ\cD^\dagger(O) = \cM^\dagger(O) = O$ given $O\in\cN^\dagger(\cL^{\rm H}_{A'})$.
\end{proof}

As illustrated in Fig.~\ref{fig:main_result}, Theorem~\ref{thm:applicability} implies that the channel $\cN$ preserves the shadow information or, equivalently, the shadow information is retrievable, if and only if the corresponding observable $O$ is in the adjoint image of $\cL^{\rm H}_{A'}$ under $\cN$.
A physical interpretation of this theorem can be given within the Heisenberg picture, where quantum states are constant while observables evolve with time~\cite{Wilde2017book}. From this perspective, Theorem~\ref{thm:applicability} is saying that $\tr[\rho O]$ can be recovered if and only if there is an observable $Q$ evolving into $O$ under the backward dynamics prescribed by the given $\cN$. This theorem also manifests the ultimate limitation of a quantum channel on preserving the shadow information queried by some observables.

\subsection{Quantifying the Level of Shadow Information Preservation}
Following Theorem~\ref{thm:applicability}, for a quantum channel $\cN_{A\to A'}$, we define the dimension of $\cN^\dagger(\cL^{\rm H}_{A'})$ as the channel $\cN$'s \textit{effective shadow dimension} $d_{\rm s}(\cN)$, which quantifies how much shadow information is preserved by $\cN$. To compute the effective shadow dimension, we note that the set of all Hermitian operators $\cL^{\rm H}_{A'}$ can be viewed as a vector space over the field of real numbers, and the set of all linear operators $\cL_{A'}$ is a vector space over the field of complex numbers. Both vector spaces are with dimension $d_{A'}^2$, where $d_{A'} = \dim(\cH_{A'})$. Note that a basis for $\cL^{\rm H}_{A'}$ is also a basis for $\cL_{A'}$, which implies that $\dim(\cN^\dagger(\cL^{\rm H}_{A'})) = \dim(\cN^\dagger(\cL_{A'}))$. Then, by the properties of the adjoint~\cite{Axler2015linear}, the image of $\cN^\dagger$ is the space orthogonal to the null space of $\cN$. Hence, the dimension of $\cN^\dagger(\cL^{\rm H}_{A'})$ is $d_{A'}^2 - \dim({\rm null}(\cN)) = \dim({\rm Image}(\cN))$, which equals the rank of a matrix of $\cN$ as a linear map.

A matrix of $\cN$ as a linear map can be obtained from the channel's Kraus representation. Given $\cN(\ketbra{i}{j}) = \sum_k E_k\ketbra{i}{j}E_k^\dagger$, we have ${\rm vec}(\cN(\ketbra{i}{j})) = \sum_k \overline{E_k}\ket{j} \ox E_k\ket{i} = (\sum_k \overline{E_k} \ox E_k) \ket{ji}$, where $\overline{E_k}$ is the complex conjugate of $E_k$. Hence, $M_\cN \equiv \sum_k \overline{E_k} \ox E_k$ is a matrix representation of $\cN$ as a linear map. The effective shadow dimension $d_{\rm s}(\cN)$ is efficiently computable with respect to the system dimension as the matrix rank of $M_\cN$.

Based on the effective shadow dimension, we define a measure called \emph{shadow destructivity}:
\begin{align}
    \zeta(\cN) \equiv \log \frac{d^2}{d_{\rm s}(\cN)},
\end{align}
where $d$ is the dimension of the Hilbert space that the input linear operators in $\cL_{A}$ act on. The shadow destructivity quantifies a quantum channel's \update{capability} to destruct shadow information and has some meaningful properties:
\begin{enumerate}
    \item (\emph{Faithfulness}) $\zeta(\cN) = 0$ if and only if all the shadow information is recoverable. In other words, $\zeta$ is strictly larger than $0$ if and only if some shadow information is irreversibly lost.

    \begin{proposition}{\rm (Faithfulness of shadow destructivity).}
    The shadow destructivity is faithful in the sense that $\zeta(\cN) \geq 0$ and it vanishes, i.e., $\zeta(\cN) = 0$, if and only if for each observable $O$, there exists a retriever $\cD$ such that $\tr[\cD\circ\cN(\rho) O] = \tr[\rho O]$ for every quantum state $\rho$.
    \end{proposition}
    
    \begin{proof}
    First suppose $\zeta(\cN) \equiv \log(d^2/d_{\rm s}(\cN)) = 0$, which is only true when $d_{\rm s}(\cN) = d^2$ and thus implying that $\cN^\dagger$ is surjective, i.e., $\cL^{\rm H}$ being the image of $\cN^\dagger$. In other words, the image of $\cN^\dagger$ consists of all observables. Then, according to Theorem~\ref{thm:applicability} and Lemma~\ref{lemma:minimum_requirement}, for every observable $O$, there is a retriever $\cD$ such that $\tr[\cD\circ\cN(\rho) O] = \tr[\rho O]$ for every state $\rho$.
    
    Now suppose that for every observable $O$, there is a retriever $\cD$ such that $\tr[\cD\circ\cN(\rho) O] = \tr[\rho O]$ for every state $\rho$. By Theorem~\ref{thm:applicability}, this implies that all observables are in the image of $\cN^\dagger$. Since the image can only be a subspace of $\cL^{\rm H}$, we know that the image is $\cL^{\rm H}$ itself, whose dimension is $d^2$. Hence, $d_{\rm s}(\cN) \equiv \dim(\cN^\dagger(\cL^{\rm H})) = d^2$, and we conclude that $\zeta(\cN) \equiv \log(d^2/d_{\rm s}(\cN)) = 0$.
    
    We have shown that $\zeta(\cN) = 0$ if and only if all the shadow information is recoverable. In addition, the shadow destructivity is non-negative by its definition. Therefore, the shadow destructivity is faithful.
    \end{proof}
    
    \item (\emph{Additivity}) The global shadow destructivity equals the sum of the shadow destructivity of local channels, that is, $\zeta(\cM\ox\cN) = \zeta(\cM) + \zeta(\cN)$ for two quantum channels $\cM$ and $\cN$ acting on different subsystems.

    \begin{proposition}{\rm (Additivity of shadow destructivity).}
    The shadow destructivity is additive with respect to tensor product, i.e.,
    \begin{align}
        \zeta(\cM\ox\cN) = \zeta(\cM) + \zeta(\cN),
    \end{align}
    where $\cM_{A\to A'}$ and $\cN_{B\to B'}$ are two quantum channels.
    \end{proposition}
    \begin{proof}
    We have showed that the effective shadow dimension of a channel equals the rank of the matrix of it as a linear map, that is: $d_{\rm s}(\cM) = {\rm rank}(\update{M_\cM})$, $d_{\rm s}(\cN) = {\rm rank}(\update{M_\cN})$, and $d_{\rm s}(\cM\ox\cN) = {\rm rank}(\update{M_{\cM\ox\cN}})$. Hence, $\zeta(\cM\ox\cN) \equiv \log((d_Ad_B)^2/d_{\rm s}(\cM\ox\cN)) = \log(d_A^2d_B^2/{\rm rank}(\update{M_{\cM\otimes\cN}}))$. It is also known that~\cite{Laub2005matrix}
    \begin{align}
    {\rm rank}(\update{M_{\cM\otimes\cN}}) \update{= {\rm rank}(M_\cM\otimes M_\cN)} = {\rm rank}(\update{M_\cM}){\rm rank}(\update{M_\cN}). 
    \end{align}
    By this property, we have 
    \begin{align}
     \zeta(\cM\ox\cN) =& \log(d_A^2d_B^2/({\rm rank}(\update{M_\cM}){\rm rank}(\update{M_\cN}))) \\
     =& \log(d_A^2/{\rm rank}(\update{M_\cM})) + \log(d_B^2/{\rm rank}(\update{M_\cN})) \\
     =& \log(d_A^2/d_{\rm s}(\cM)) + \log(d_B^2/d_{\rm s}(\cN)) = \zeta(\cM) + \zeta(\cN).
    \end{align}
    \end{proof}

    \item (\emph{Data processing inequality}) The shadow destructivity cannot increase when a quantum state is sent through more channels, that is, $\zeta(\cN\circ\cM) \geq \max(\zeta(\cN),\zeta(\cM))$. This property implies that we can only lose shadow information by sending quantum states through quantum channels.

    \begin{proposition}{\rm (Data processing inequality of shadow destructivity).}
    For two quantum channels $\cM$ and $\cN$, it holds that $\zeta(\cN\circ\cM) \geq \max(\zeta(\cN),\zeta(\cM))$.
    \end{proposition}
    \begin{proof}
    \update{If we write the channels $\cM$ and $\cN$ in their matrix representations $M_\cM$ and $M_\cN$, the the matrix representation of the composite channel $\cN\circ\cM$ is $M_\cN M_\cM$, which is known to have a rank ${\rm rank}(M_\cN M_\cM) \leq \min({\rm rank}(M_\cN), {\rm rank}(M_\cM))$. This is equivalent to $d_{\rm s}(\cN\circ\cM) \leq \min(d_{\rm s}(\cN), d_{\rm s}(\cM))$, which implies $\zeta(\cN\circ\cM) \geq \max(\zeta(\cN),\zeta(\cM))$.}
    \end{proof}
\end{enumerate}

As the shadow destructivity derived from Theorem~\ref{thm:applicability} can be computed from the matrix rank of $M_\cN$, Theorem~\ref{thm:applicability} endows the matrix rank of $M_\cN$ an operational meaning in shadow information recoverability, where a higher rank indicates that more shadow information is recoverable. It is well known that a matrix is invertible if and only if it has full rank. Since full rank corresponds to the total recoverability of shadow information, invertibility of a channel is equivalent to its total shadow information recoverability.

\subsection{Quantifying the Cost of Retrieving Shadow Information}
By saying that a quantum channel preserves some shadow information, we mean this information can be retrieved after a state is corrupted by this channel. Though the effective shadow dimension of a quantum channel captures its overall capability of preserving shadow information, a piece of shadow information being preserved does not mean that retrieving it is cost-free, and different pieces of information can have different retrieving cost. In the following part, we derive a measure called \textit{retrieving cost} from the number of copies of the corrupted state required to estimate the desired shadow information within an acceptable accuracy.

Note that the retriever $\cD$ is an HPTS map but not necessarily CPTP, i.e., not necessarily a quantum channel. While such a map $\cD$ is not physically implementable, we can simulate its action by decomposing it as a linear combination of CPTP maps: $\cD = \sum_j c_j\cD_j$, and utilizing Eq.~\eqref{eq:linear_comb_of_qchannel}. As Lemma~\ref{lemma:hpts} implies, a linear map is physically simulatable in this framework if and only if it is HPTS.

\begin{lemma}\label{lemma:hpts}
A linear map $\cD$ is HPTS if and only if it can be written as a linear combination of quantum channels $\{\cD_j\}$ with real numbers $\{c_j\}$: $\cD = \sum_j c_j\cD_j$.
Furthermore, for any HPTS map $\cD$, there exist two quantum channels $\cD_1, \cD_2$ with real numbers $c_1, c_2$ such that 
\begin{equation}\label{eq:retriever_2_terms}
    \cD = c_1\cD_1 + c_2\cD_2.
\end{equation}
\end{lemma}
\begin{proof}
We first prove that a linear combination of quantum channels is HPTS. The \Choi matrix of such a linear combination $\cD = \sum_j c_j\cD_j$ is $J_\cD = \sum_j c_jJ_{\cD_j}$, where $J_{\cD_j}\geq0$ and $\tr_B[J_{\cD_j}] = I_A$ for each $j$. As a linear combination of Hermitian operators is also Hermitian, $\cD$ is Hermitian-preserving. Also noting that $\tr_B[J_\cD] = \sum_j c_j\tr_B[J_{\cD_j}] = \sum_j c_j I_A$, we have $\tr_B[J_\cD] = c I_A$, where $c = \sum_j c_j$. Hence, $\cD$ is HPTS.

For the other direction, if a map $\cD$ is HPTS, its \Choi matrix $J_\cD$ satisfies $J_\cD$ being Hermitian and $\tr_B[J_\cD] = c I_A$ for some real coefficient $c$. If $c\neq0$, $\cD/c$ is an HPTP map and there exist two CPTP maps $\cD_1,\cD_2$ and two real numbers $c_1\geq0, c_2\leq0$ such that $\cD/c = c_1\cD_1 + c_2\cD_2$~\cite{Jiang2021physical}. Then we have $\cD = c\cdot c_1\cD_1 + c\cdot c_2\cD_2$. On the other hand, if $c=0$, $\cD + \id$ is an HPTP map and there still exist two CPTP maps $\cD_1,\cD_2$ and two real numbers $c_1\geq0, c_2\leq0$ such that $\cD + \id = c_1\cD_1 + c_2\cD_2$. Hence, $\cD$ can be written as $\cD = c_1\cD_1 + (c_2-1)\cD_2'$, where $\cD_2' = c_2\cD_2/(c_2-1) - \id/(c_2-1)$ is a quantum channel.
\end{proof}

Now let $\cD$ be an HPTS map whose decomposition is $\sum_j c_j\cD_j$.
With this decomposition, we can simulate its action on the expectation value by the probabilistic sampling method prescribed by Eq.~\eqref{eq:linear_comb_of_qchannel}. Specifically, in the $s$-th round of total S times of sampling, we first sample a quantum channel $\cD^{(s)}$ from $\{\cD_j\}$ with probabilities $\{|c_j|/\gamma\}$, where $\gamma = \sum_j |c_j|$, and apply it to a copy of the corrupted state to obtain $\cD^{(s)}\circ\cN(\rho)$. Then, we measure this state in the eigenbasis $\{\proj{\psi_j}\}_j$ of the given observable $O = \sum_j o_j\proj{\psi_j}$, where $o_j \in [-1,1]$, and obtain the measurement value $o^{(s)}$. After $S$ rounds of sampling, we attain an estimation for the expectation value $\tr[\rho O]$ as
\begin{align}\label{eq:expectation_value}
    \xi = \frac{\gamma}{S} \sum_{s=1}^S \sgn(c^{(s)})o^{(s)}.
\end{align}
 By the Hoeffding's inequality~\cite{hoeffding1994probability}, the number of rounds required to obtain the estimation within an error $\varepsilon$ with a probability no less than $1-\delta$ is 
\begin{equation}\label{eq:sampling_times}
    S = 2\gamma^2 \log(2/\delta) / \varepsilon^2.
\end{equation}

It can be seen that the required number of copies $S$ of the corrupted state is directly related to $\gamma$, the sum of absolute values of all $c_j$. Hence, $\gamma$ can be used to characterize the cost of simulating the HPTS map $\cD$. It is desirable to find a decomposition of $\cD$ making $\gamma$ as small as possible, and we denote the minimum possible value of $\gamma$ as $\gamma_{\min}(\cD)$. In Ref.~\cite{Jiang2021physical}, the authors define the logarithm of $\gamma_{\min}(\cD)$ as the physical implementability of an HPTP map $\cD$.

Naturally, given a channel $\cN$, we define the \emph{retrieving cost} with respect to an observable $O$ to be
\begin{align}
    \gamma_O(\cN) = \min \{\gamma_{\min}(\cD) \vert \cD\text{ is HPTS},~\cN^\dagger\circ\cD^\dagger(O) = O\},
\end{align}
which is the $\gamma_{\min}(\cD)$ minimized over $\cD$ that can retrieve the shadow information queried by $O$ from the noisy channel $\cN$.

Notably, Lemma~\ref{lemma:hpts} states that any HPTS map can be decomposed as a linear combination of two quantum channels (Eq. \eqref{eq:retriever_2_terms}), and we show in the Supplemental Material that the minimum retrieving cost can be achieved by such a decomposition. This is true because all channels with non-negative coefficients can be grouped into a single channel with a non-negative coefficient, and we can do the same thing for all channels with negative coefficients. The existence of such an optimal decomposition with two CPTP maps leads to an efficient way to compute this measure through a semidefinite program (SDP) in terms of linear maps' \Choi matrices:
\begin{subequations}\label{eq:sdp_total}
\begin{align}
\gamma_{O}(\cN) = \min &\; c_1 + c_2 \\
\text{s.t.} &\; J_\cD \equiv J_{\cD_1} - J_{\cD_2} \\
            &\; J_{\cD_1} \geq 0,\; J_{\cD_2} \geq 0 \label{eq:sdp_cp} \\ 
            &\; \tr_{A_o'}[J_{\cD_{1{A'_iA'_o}}}] = c_1I_{A'_i},\; \tr_{A'_o}[J_{\cD_{2{A'_iA'_o}}}] = c_2I_{A'_i} \label{eq:sdp_tp} \\
            &\; J_{\cM_{A_iA'_o}} \equiv \tr_{A_oA'_i}[(J_{\cN_{A_oA'_i}}^T\ox I_{A_iA'_o})(J_{\id_{A_iA_o}}\ox J_{\cD_{A'_iA'_o}})] \\
            &\; \tr_{A'_o}[(I_{A_i}\ox O_{A'_o}^T) J_{\cM_{A_iA'_o}}^T] = O_{A_i}, \label{eq:sdp_min_req}
\end{align}
\end{subequations}
where $J_\id$ is the \Choi matrix of an identity channel.
Eq.~\eqref{eq:sdp_cp} corresponds to the condition that $\cD_1$ and $\cD_2$ are CP, and Eq.~\eqref{eq:sdp_tp} requires them to be TP. Eq.~\eqref{eq:sdp_min_req} corresponds to the minimum requirement for a valid retriever as given in Eq.~\eqref{eq:min_req_on_retriever}.
Note that we can establish approximate versions of Eq.~\eqref{eq:sdp_total} by relaxing the requirement in Eq.~\eqref{eq:sdp_min_req} such that $\tr_{A'_o}[(I_{A_i}\ox O_{A'_o}^T) J_{\cM_{A_iA'_o}}^T]$ is close to $O$ within a certain tolerance.
This relaxation provides a remedy to cases where a perfect retriever does not exist as well as a trade-off between the retrieving cost and the estimation precision.

\subsection{Case Study and Applications in Error Mitigation}
To intuitively understand effective shadow dimension and shadow destructivity, we study two non-invertible channels as examples: $\cN_1(\cdot) = \frac{1}{2}I(\cdot)I + \frac{1}{2}X(\cdot)X$ and $\cN_2(\cdot) = \frac{1}{2}I(\cdot)I + \frac{1}{4}X(\cdot)X + \frac{1}{4}Y(\cdot)Y$. For a quantum state
$\rho =
\begin{pmatrix}
a & b+ci \\ b-ci & 1-a
\end{pmatrix}
$, where $a, b, c$ are real numbers and $i\equiv\sqrt{-1}$, it can be easily verified that
$\cN_1(\rho) =
\begin{pmatrix}
1/2 & b \\ b & 1/2
\end{pmatrix}
$ and
$\cN_2(\rho) =
\begin{pmatrix}
1/2 & (b+ci)/2 \\ (b-ci)/2 & 1/2
\end{pmatrix}
$.
While both channels are non-invertible, it is intuitive to see that channel $\cN_2$ preserves more information than channel $\cN_1$ since the imaginary part of off-diagonal elements is preserved. This intuition coincides with the effective shadow dimension of these two channels, $d_{\rm s}(\cN_2) = 3 > d_{\rm s}(\cN_1) = 2$, which implies that $\cN_2$ preserves one more shadow dimension than $\cN_1$. The shadow destructivity of these channels, $\zeta(\cN_1) = \log (4/2) = 1$ and $\zeta(\cN_2) = \log (4/3) \approx 0.415$, has similar implications.

Now recall the \textit{retrieving cost}, which
quantifies the resources required to recover a certain piece of shadow information. The corresponding optimal retrieving protocol can be considered as a method to mitigate errors, endowing the retrieving cost with a practical meaning in quantum error mitigation~\cite{Temme2017,Endo2018a,Takagi2020,Jiang2021physical,Piveteau2021,Bravyi2020,Maciejewski2020,Czarnik2020,McClean2019,Cai2021,Bultrini2021,Cai2021b,Czarnik2021,Piveteau2021,Lostaglio2021,Takagi2021,Koczor2020,Endo2020,Wang2021a,Suzuki2020}.

\update{In error mitigation, PEC methods are promising to suppress noises in quantum circuits run on near-term quantum computers and return us unbiased estimation of expectation values~\cite{Temme2017,Endo2018a,Takagi2020,Jiang2021physical}. Briefly speaking, the idea of PEC is to decompose the inverse (HPTP) of a noisy channel $\cN$ into a linear combination of physically implementable channels (CPTP), i.e., $\cN^{-1} = \sum_j a_j \cA_j$, with real coefficients $a_j$ and CPTP maps $\cA_j$, where $\sum_j a_j = 1$. The inverse map $\cN^{-1}$ is implemented by sampling the channels $\cA_j$ according to a probability distribution $p(j) = \frac{|a_j|}{\gamma_{\rm con}}$, where $\gamma_{\rm con} = \sum_j |a_j|$. In this way, the unbiased estimation of the expectation value of an observable $O$ with respect to a state $\rho$ is obtained by manipulating multiple copies of the noisy state $\cN(\rho)$:
\begin{equation}
    \gamma_{\rm con}\sum_j p(j)\tr[\cA_j\circ \cN(\rho)O] = \tr[\cN^{-1}\circ \cN(\rho)O] = \tr[\rho O].
\end{equation}
The quantity $\gamma_{\rm con}$ is known as the overhead (also called retrieving cost in the setting of our work) because it captures how many copies of the noisy state are needed to ensure the estimation within the acceptable precision. Therefore, $\gamma_{\rm con}$ is usually used to evaluate the efficiency of an error mitigation protocol.}

\update{Our proposed method can be considered as an approach to mitigating errors, where the retrieving cost corresponds to the overhead of an error mitigation protocol. For a given noisy channel $\cN$ and an observable $O$ of interest, the SDP in Eq.~\eqref{eq:sdp_total} can give us a retriever $\cD$ and its decomposition $\cD = c_1\cD_1 + c_2\cD_2$ with the optimized retrieving cost $\gamma_{\rm pro} = |c_1|+|c_2|$.
The expectation of the estimation made by our method is
\begin{align}
    \gamma_{\rm pro} \left(\frac{|c_1|}{\gamma_{\rm pro}}\tr[\cD_1\circ \cN(\rho)O]+\frac{|c_2|}{\gamma_{\rm pro}}\tr[\cD_2\circ \cN(\rho)O]\right)&= \tr[\cD\circ \cN(\rho)O] \\
    &= \tr[\rho \cN^\dagger\circ\cD^\dagger(O)] \\
    &= \tr[\rho O].
\end{align}
Note that in our method, the retriever $\cD$ is not necessarily the inverse of the noisy channel, i.e., $\cD \neq \cN^{-1}$.}

In the following, we compare the retrieving cost between our proposed protocol $\gamma_{\rm pro}$ and the conventional PEC method $\gamma_{\rm con}$ ~(see, e.g., ~\cite{Temme2017,Endo2018a,Takagi2020,Jiang2021physical,Piveteau2021}), with the assumption that full knowledge of the noise is accessible and the noise is modeled as a channel just before the measurement.

As a concrete example, consider retrieving shadow information $\tr[\rho X]$ from a state $\rho$ corrupted by a GAD channel~\cite{Chirolli2018}, where $X$ is the Pauli $X$ operator. The cost of our method is $\gamma_{\rm pro}=\frac{1}{\sqrt{1-\epsilon}}$, while the cost of the conventional method is $\gamma_{\rm con} = \frac{|1-2p|\epsilon+1}{1-\epsilon}$~\cite{Jiang2021physical}, where $\epsilon$ is the damping factor and $p$ is the temperature indicator associated with the GAD channel. It is obvious to see that $\gamma_{\rm pro} < \gamma_{\rm con}$ for any $0<\epsilon<1$, as shown in Fig. \ref{fig:AD_cost}. Detailed protocols are provided in the Supplemental Material.

\begin{figure}[t]
    \centering
    \includegraphics[width=0.66\textwidth]{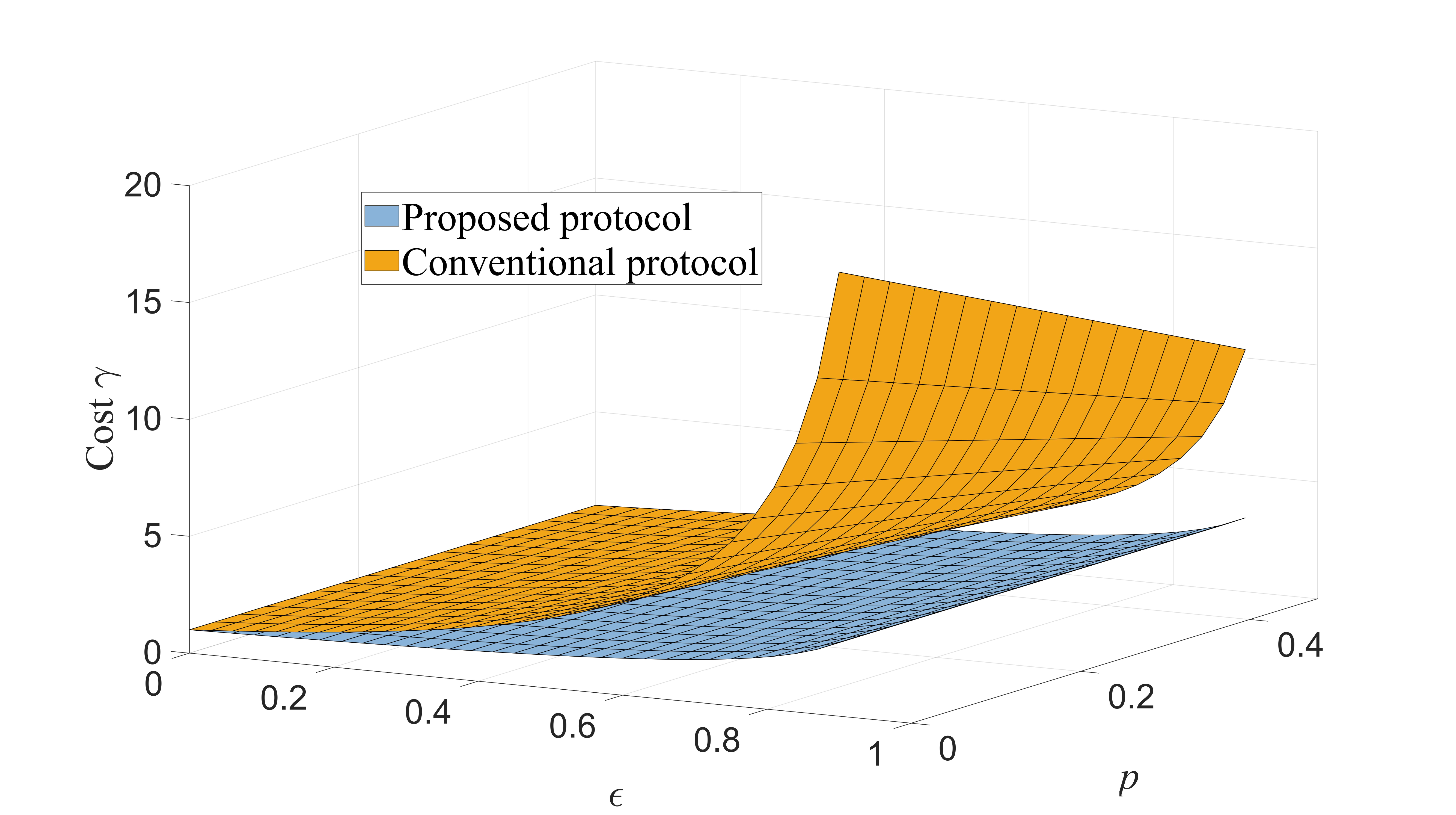}
    \caption{Comparison of the retrieving cost between our proposed protocol and conventional PEC protocol in estimating expectation value $\tr[\rho O]$ from noisy state $\cN(\rho)$. The observable is $O=X$, and $\cN$ is the GAD channel with coefficients $p\in [0, 0.5]$ and $\epsilon \in [0, 0.9]$.}
    \label{fig:AD_cost}
\end{figure}

Now we apply our protocol to ground state energy estimation with VQE.
VQE~\cite{Peruzzo2014} is a promising near-term algorithm for estimating the ground state energy of a Hamiltonian $H$. It aims to minimize the cost function $E = \tr[\proj{\psi(\bm{\theta})} H]$ by tuning parameters $\bm{\theta}$ in a parameterized quantum circuit used to prepare the ansatz state $\ket{\psi(\bm{\theta})}$. In practice, $H$ is decomposed into a linear combination of tensor products of Pauli tensors $\{O_j\}$, i.e., $H=\sum_j h_j O_j$, where each $h_j$ is a real coefficient. The core subroutine of VQE is estimating the expectation value of each term $\tr[\rho O_j]$.

We use Eq.~\eqref{eq:sampling_times} to estimate the numbers of sampling rounds required by VQE to estimate the ground state energies of given molecules when there is depolarizing noise~\cite{Nielsen2010} on each qubit. Table~\ref{table:vqe_cost_comparison} shows the Comparison between our protocol and the conventional protocol \cite{Jiang2021physical}. The retrieving cost of our proposed method is $\gamma_{\rm pro}=\frac{1}{1-\epsilon}$ (see Supplemental Material), while the conventional method is $\gamma_{\rm con}=\frac{1+(1-2/d^2)\epsilon}{1-\epsilon}$, where $\epsilon$ is the noise level, and $d$ is the dimension of the system. It is clear that the number of sampling rounds of our method is smaller than that of the conventional method.

\newcommand{\ra}[1]{\renewcommand{\arraystretch}{#1}}
\setlength{\tabcolsep}{5pt}
\begin{table}[h]
\centering
\ra{1.4}
\begin{tabular}{@{}lccc@{}}
\toprule[1pt]
&  \multicolumn{3}{c}{Molecule} \\
\cmidrule[0.8pt]{2-4}
& H$_2$ & HF & CO$_2$  \\
\midrule[0.5pt]
Proposed &  $1.24\times 10^5$ & $7.84\times 10^{10}$ & $1.28\times 10^{13}$\\
Conventional &  $1.60\times 10^5$ & $1.79\times 10^{11}$ & $1.21\times 10^{14}$\\
\bottomrule[1pt]
\end{tabular}
\caption{Sampling times for estimating the ground state energies of three different molecules under the depolarizing noise with a noise level $\epsilon=0.1$. The precision parameters in Eq.~\eqref{eq:sampling_times} are set to $\varepsilon=\delta=0.01$. For estimating the sampling times, each coefficient $h_j$ of a Hamiltonian $H=\sum_j h_jO_j$ is assumed to be $\max \{|h_j|\}$, and the term where $O_j$ is the identity operator is excluded.}
\label{table:vqe_cost_comparison}
\end{table}

Besides the advantage of having a lower cost, our method has a larger range of applicability. The conventional method implements the retriever as the inverse of the noisy channel. In contrast, it only works for the cases where the noise is invertible. However, our method, employing an observable-adaptive strategy, also works for some cases where the noise is non-invertible. As long as the observable of interest is in the image of the noise's adjoint map, our method can provide an optimal protocol for recovering the corresponding shadow information.

\section{Conclusions}
In this work, we establish two measures from an operational perspective to quantify a quantum channel's influence on shadow information encoded in quantum states, answering the questions posed at the beginning.
Our work delivers a systematic framework of quantifying the information preservativity and destructivity of a quantum channel. Meanwhile it establishes an optimal way of extracting noiseless classical information from noisy states, which could be useful for near-term quantum information processing tasks.

While the shadow destructivity and the retrieving cost capture different aspects of shadow information recoverability, a single unified measure that captures both aspects may be worth further studies.
For future work, it would be interesting to see how the techniques presented in this work can be combined with quantum error correction~\cite{Gottesman2005,Lidar2013,Terhal2015} or quantum error mitigation~\cite{Koczor2020,Cai2021b,Endo2020,Wang2021a,Suzuki2020}. We also expect that our ideas could be applied to near-term quantum tasks or applications on noisy quantum devices~\cite{Preskill2018}.

\section*{Acknowledgements}
X. Z. and B. Z. contributed equally to this work.
We would like to thank Chenghong Zhu for helpful discussions. \update{We would also like to thank the referees for their helpful comments, which greatly improved the readability of the manuscript.}
Part of this work was done when X. Z., B. Z., and Z. X. were research interns at Baidu Research.
\bibliographystyle{unsrtnat}
\bibliography{Em}

\clearpage
\setcounter{subsection}{0}
\setcounter{table}{0}
\setcounter{figure}{0}

\vspace{2cm}
\vspace{2cm}
\begin{center}
{\textbf{\large Supplemental Material for \\Information Recoverability of Noisy Quantum States}}
\end{center}

\renewcommand{\theequation}{S\arabic{equation}}
\renewcommand{\theproposition}{S\arabic{proposition}}
\renewcommand{\theHequation}{Supplement.\theequation}
\renewcommand{\theHproposition}{Supplement.\theproposition}
\setcounter{equation}{0}
\setcounter{section}{0}
\setcounter{proposition}{0}
\section{Adjoint of a Trace-scaling Map}
\renewcommand{\theproposition}{S\arabic{proposition}}
\begin{lemma}\label{lemma:hpts_adjoint}
The adjoint of a trace-scaling map is unit-scaling,
and vice versa. By calling a linear map $\cN$ unit-scaling, we mean it scales the identity operator, i.e., $\cN(I) = pI$ for some real number $p$.
\end{lemma}
\begin{proof}
Suppose that $\cN$ is a trace-scaling map such that $\tr[\cN(\cdot)] = p\tr[\cdot]$ for some real number $p$. If $p\neq0$, we can define a TP map $\cM\equiv\cN/p$. It is known that the adjoint of TP maps are unital~\cite{khatri2020principles}, which preserves the identity operator, and vice versa. Thus, $\cM^\dagger$, the adjoint of $\cM$, is a unital map. Then, $\cN^\dagger = p\cM^\dagger$, the adjoint of the trace-scaling map $\cN$, is a unit-scaling map. On the other hand, if $p=0$, then $\cN+\id$ is TP and thus $\cN^\dagger+\id$ is unital, making $\cN^\dagger$ a unit-scaling map. Hence, the adjoint of a trace-scaling map is unit-scaling.

Now suppose that $\cN$ is a unit-scaling map such that $\cN(I) = pI$ for some real number $p$. If $p\neq0$, then $\cM\equiv\cN/p$ is a unital map. Thus, we know its adjoint $\cM^\dagger$ is TP, making $\cN^\dagger = p\cM^\dagger$ trace-scaling. On the other hand, if $p=0$, then $\cN+\id$ is unital and hence its adjoint $\cN^\dagger+\id$ is TP, again making $\cN^\dagger$ trace-scaling. Hence, the adjoint of a unit-scaling map is trace-scaling.
\end{proof}

\section{Retrieving Cost}
Recall that, given a quantum channel $\cN$, the shadow retrieving cost with respect to an observable $O$ is defined as
\begin{align}
    \gamma_O(\cN) = \min \{\gamma_{\min}(\cD) \vert \cD\text{ is HPTS},~\cN^\dagger\circ\cD^\dagger(O) = O\},
\end{align}
where $\gamma_{\min}(\cD)$ is the minimum cost for simulating $\cD$. Here, we give a proof that the minimum cost of simulating an HPTS map can be achieved by a decomposition with two CPTP maps. The same thing has been proved for HPTP maps in Ref.~\cite{Jiang2021physical}, and the proof here is similar.

\renewcommand{\theproposition}{S\arabic{proposition}}
\begin{proposition}
For any HPTS map $\cD$,
\begin{align}
    \gamma_{\min}(\cD) = \min \left\{|c_1| + |c_2| \big\vert \cD = c_1\cD_1 + c_2\cD_2;~\cD_1,\cD_2\text{ are CPTP};~c_1,c_2\in \mathbb{R} \right\}.
\end{align}
\end{proposition}
\begin{proof}
Lemma~\ref{lemma:hpts} implies that $\gamma_{\min}(\cD)$ is finite as $\cD$ can be decomposed into a linear combination of two CPTP maps. Suppose $\gamma_{\min}(\cD)$ is achieved by $\cD = \sum_{j=1}^T c'_j\cD'_j$, where $T\geq3$ so that $\gamma_{\min}(\cD) = \sum_{j=1}^T |c'_j|$. Then, it is always possible to construct two CPTP maps $\cD_1,\cD_2$ so that $\cD = c_1\cD_1 + c_2\cD_2$ and $|c_1|+|c_2| = \gamma_{\min}(\cD)$. In particular, let $c_1 = \sum_{j:c'_j\geq0} c'_j$, $\cD_1 = \sum_{j:c'_j\geq0} c'_j/c_1\cdot \cD'_j$, $c_2 = \sum_{j:c'_j<0} c'_j$, and $\cD_2 = \sum_{j:c'_j<0} c'_j/c_2\cdot \cD'_j$. One can check that $\cD_1$ and $\cD_2$ are CPTP and $c_1\cD_1 + c_2\cD_2 = \sum_{j=1}^T c'_j\cD'_j$. Thus, this decomposition of $\cD$ is valid. Plus, $|c_1| + |c_2| = \sum_{j=1}^T |c'_j|$, ensuring that this decomposition achieves $\gamma_{\min}(\cD)$.
\end{proof}

\subsection{Dual SDP for Retrieving Cost}\label{sec:dual_sdp}
For a quantum channel $\cN_{A\to A'}$, the primal SDP characterization of its retrieving cost with respect to an observable $O$ is
\begin{subequations}\label{eq:primal SDP}
\begin{align}
\gamma_{O}(\cN) = \min &\; c_1 + c_2 \\
\text{s.t.} &\; J_\cD \equiv J_{\cD_1} - J_{\cD_2} \\
            &\; J_{\cD_1} \geq 0,\; J_{\cD_2} \geq 0 \\ 
            &\; \tr_{A_o'}[J_{\cD_{1{A'_iA'_o}}}] = c_1I_{A'_i},\; \tr_{A'_o}[J_{\cD_{2{A'_iA'_o}}}] = c_2I_{A'_i} \\
            &\; J_{\cM_{A_iA'_o}} \equiv \tr_{A_oA'_i}[(J_{\cN_{A_oA'_i}}^T\ox I_{A_iA'_o})(J_{\id_{A_iA_o}}\ox J_{\cD_{A'_iA'_o}})] \\
            &\; \tr_{A'_o}[(I_{A_i}\ox O_{A'_o}^T) J_{\cM_{A_iA'_o}}^T] = O_{A_i}.
\end{align}
\end{subequations}
The Lagrange function is
\begin{align}
    L(J_{\cD_1}, J_{\cD_2}, c_1, c_2, M, N, K) &= c_1 + c_2 + \langle M, \tr_{A'_o}[J_{\cD_1}] - c_1I_{A'_i} \rangle + \langle N, \tr_{A'_o}[J_{\cD_2}] - c_2I_{A'_i} \rangle \nonumber\\
    &~\quad + \langle K, \tr_{A'_o}[(I_{A_i}\ox O_{A'_o}^T) J_{\cM}^T] - O\rangle\\
    &= c_1(1 - \tr[M]) + c_2(1 - \tr[N]) + \langle M\ox I, J_{\cD_1} \rangle \nonumber\\
    &~\quad + \langle N\ox I, J_{\cD_2} \rangle - \tr[KO] + \langle K, \tr_{A'_o}[(I_{A_i}\ox O_{A'_o}^T) J_{\cM}^T]\rangle
\end{align}
where $M$, $N$, $K$ are the dual variables. The last term can be expanded as
\begin{align}
    \langle K_{A_i}, \tr_{A'_o}[(I_{A_i}\ox O_{A'_o}^T) J_{\cM}^T]\rangle &= \tr[K_{A_i} \tr_{A'_o}[(I_{A_i}\ox O_{A'_o}^T) J_{\cM_{A_iA'_o}}^T]]\\
    &= \tr[(K_{A_i}\ox O_{A'_o}^T) J_{\cM_{A_iA'_o}}^T]\\
    &= \tr[(K_{A_i}^T\ox O_{A'_o}) J_{\cM_{A_iA'_o}}]\\
    &= \tr[(K_{A_i}^T\ox O_{A'_o}) \tr_{A_oA'_i}[(J_{\cN_{A_oA'_i}}^T\ox I_{A_iA'_o})(J_{\id_{A_iA_o}}\ox J_{\cD_{A'_iA'_o}})]]\\
    &= \tr[(K_{A_i}^T\ox J_{\cN_{A_oA'_i}}^T\ox O_{A'_o}) (J_{\id_{A_iA_o}}\ox J_{\cD_{A'_iA'_o}})]\\
    &= \tr[(K_{A_i}^T\ox J_{\cN_{A_oA'_i}}^T\ox O_{A'_o}) (J_{\id_{A_iA_o}}\ox I_{A'_iA'_o}) (I_{A_iA_o}\ox J_{\cD_{A'_iA'_o}})]\\
    &= \tr[\tr_{A_iA_o}[(K_{A_i}^T\ox J_{\cN_{A_oA'_i}}^T\ox O_{A'_o}) (J_{\id_{A_iA_o}}\ox I_{A'_iA'_o})] J_{\cD_{A'_iA'_o}}]\\
    &= \langle \tr_{A_iA_o}[(K_{A_i}^T\ox J_{\cN_{A_oA'_i}}^T\ox O_{A'_o}) (J_{\id_{A_iA_o}}\ox I_{A'_iA'_o})], J_{\cD_{1 A'_iA'_o}} \rangle \nonumber\\
    &~\quad - \langle \tr_{A_iA_o}[(K_{A_i}^T\ox J_{\cN_{A_oA'_i}}^T\ox O_{A'_o}) (J_{\id_{A_iA_o}}\ox I_{A'_iA'_o})], J_{\cD_{2 A'_iA'_o}} \rangle
\end{align}
Thus, the Lagrange function can be expressed as
\begin{align}
    L(J_{\cD_1}, J_{\cD_2}, c_1, c_2, M, N, K) &= c_1(1 - \tr[M]) + c_2(1 - \tr[N]) - \tr[KO] \nonumber\\
    &~\quad + \langle M\ox I + \tr_{A_iA_o}[(K_{A_i}^T\ox J_{\cN_{A_oA'_i}}^T\ox O_{A'_o}) (J_{\id_{A_iA_o}}\ox I_{A'_iA'_o})], J_{\cD_1} \rangle \nonumber\\
    &~\quad + \langle N\ox I - \tr_{A_iA_o}[(K_{A_i}^T\ox J_{\cN_{A_oA'_i}}^T\ox O_{A'_o}) (J_{\id_{A_iA_o}}\ox I_{A'_iA'_o})], J_{\cD_2} \rangle
\end{align}
The corresponding Lagrange dual function is
\begin{align}
    g(M, N, K) = \inf_{J_{\cD_1}\geq0, J_{\cD_2}\geq0} L(J_{\cD_1}, J_{\cD_2}, M, N, K)
\end{align}
For $J_{\cD_1}\geq0$ and $J_{\cD_2}\geq0$, it must hold that $\tr[M]\leq1, \tr[N]\leq1$,
\begin{align}
    M_{A'_i}\ox I_{A'_o} + \tr_{A_iA_o}[(K_{A_i}^T\ox J_{\cN_{A_oA'_i}}^T\ox O_{A'_o}) (J_{\id_{A_iA_o}}\ox I_{A'_iA'_o})] &\geq 0, \text{ and}\\
    N_{A'_i}\ox I_{A'_o} - \tr_{A_iA_o}[(K_{A_i}^T\ox J_{\cN_{A_oA'_i}}^T\ox O_{A'_o}) (J_{\id_{A_iA_o}}\ox I_{A'_iA'_o})] &\geq 0.
\end{align}
Thus, we arrive at the following dual SDP:
\begin{subequations}\label{eq:dual_sdp}
\begin{align}
\gamma_{O}(\cN) = \max &\; - \tr[KO] \\
\text{s.t.} &\; \tr[M]\leq1,\; \tr[N]\leq1\\
            &\; M_{A'_i}\ox I_{A'_o} + \tr_{A_iA_o}[(K_{A_i}^T\ox J_{\cN_{A_oA'_i}}^T\ox O_{A'_o}) (J_{\id_{A_iA_o}}\ox I_{A'_iA'_o})] \geq 0\\
            &\; N_{A'_i}\ox I_{A'_o} - \tr_{A_iA_o}[(K_{A_i}^T\ox J_{\cN_{A_oA'_i}}^T\ox O_{A'_o}) (J_{\id_{A_iA_o}}\ox I_{A'_iA'_o})] \geq 0.
\end{align}
\end{subequations}

\section{Retrieving Cost for GAD channel}\label{sec:gad_channel}
A GAD channel can be used to describe the energy dissipation to the environment at finite temperature~\cite{Chirolli2018}. It is one of the realistic sources of noise in superconducting quantum computing. For single-qubit cases, a GAD channel can be characterized by the following Kraus operators:
\begin{align}
    E_0=\sqrt{p}\left(\begin{array}{cc}
       1  & 0 \\
        0 & \sqrt{1-\epsilon}
    \end{array}\right)&, E_1=\sqrt{p}\left(\begin{array}{cc}
       0  & \sqrt{\epsilon} \\
        0 & 0
    \end{array}\right),\\
    E_2=\sqrt{1-p}\left(\begin{array}{cc}
       \sqrt{1-\epsilon}  & 0 \\
        0 & 1
    \end{array}\right)&,
    E_3=\sqrt{1-p}\left(\begin{array}{cc}
       0  & 0 \\
        \sqrt{\epsilon} & 0
    \end{array}\right),
\end{align} 
where $\epsilon$ is the damping factor and $p$ is the indicator of the temperature of the environment. A quantum state $\rho$ after going through the GAD channel is given by $\cN_{\rm GAD}(\rho)=\sum_{i=0}^3 E_i\rho E_i^\dagger$. 
Note that the amplitude damping channel is a special case of the GAD channel when $p=1$.
A single qubit state
$\rho =
\begin{pmatrix}
\rho_{00} & \rho_{01} \\ \rho_{10} & \rho_{11}
\end{pmatrix}
$
after going through the GAD channel is 
\begin{align}\label{seq:gad_act_on_state}
    \rho'=\cN_{\rm GAD}(\rho) &= \left(\begin{array}{cc}
       (1-\epsilon)\rho_{00}+p\epsilon(\rho_{00}+\rho_{11})  & \sqrt{1-\epsilon}\rho_{01} \\
       \sqrt{1-\epsilon}\rho_{10}  & \rho_{11}+\epsilon\rho_{00}-p\epsilon(\rho_{00}+\rho_{11})
    \end{array}\right),
\end{align}
where $\epsilon$ and $p$ are noise parameters.

\begin{proposition}
Given an observable $O\in\{X,Y\}$ and a GAD channel $\cN$, the minimum cost $\gamma$ to retrieve the information $\tr[\rho O]$ is $\gamma_O(\cN) =\frac{1}{\sqrt{1-\epsilon}}$, and the corresponding retriever $\cD$ can be written in the form of a Choi matrix
$$
J_{\cD} = c_1 J_{\cD_1} + c_2 J_{\cD_2},
$$
where $c_1 = -c_2 = \frac{1}{2\sqrt{1-\epsilon}}$, $J_{\mathcal{D}_1}  = \frac{1}{2}O^T\otimes O+\frac{1}{2}I\otimes I$, and $J_{\mathcal{D}_2} = -\frac{1}{2}O^T\otimes O+\frac{1}{2}I\otimes I$.
\end{proposition}

\begin{proof}
First, we are going to prove  $\gamma_O(\cN)\le\frac{1}{\sqrt{1-\epsilon}}$ using SDP~\eqref{eq:primal SDP}.
We show that the retriever $\cD$ above is a feasible solution with a cost of $\frac{1}{\sqrt{1-\epsilon}}$. So \update{to} be specific, we have
\begin{align}
    \cD\circ\cN(\rho) &= \cD(\rho')\\
    &= \tr_A[(\rho'^{T}\ox I)\cdot J_{\cD}] \\
    &= \frac{1}{2\sqrt{1-\epsilon}}\tr_A[(\rho'^{T}\ox I)(O^T\ox O)]\\
    &= \frac{1}{2\sqrt{1-\epsilon}} \tr_A[\rho'^T O^T\ox IO]\\
    &= \frac{1}{2\sqrt{1-\epsilon}} O \tr[\rho' O].
\end{align}
This means that
\begin{align}
    \tr[\cD \circ \cN(\rho)O] &= \frac{1}{2\sqrt{1-\epsilon}} \tr[O^2]\tr[\rho' O]\\
    &= \frac{1}{\sqrt{1-\epsilon}} \tr[\rho' O] \\
    &= \tr[O\rho],
\end{align}
where the second equality follows from the fact that $O^2 = I$ for $O\in\{X,Y\}$\update{, and the third equality can be verified with direct calculation using Eq.~\eqref{seq:gad_act_on_state}}.
Here, we have proven that the retriever $\cD$ is a feasible solution to retrieve information with cost $\frac{1}{\sqrt{1-\epsilon}}$, implying that $\gamma_O(\cN) \leq |c_1|+|c_2| = \frac{1}{\sqrt{1-\epsilon}}$.

Second, we use the dual SDP~\eqref{eq:dual_sdp} to show the cost $\gamma_O(\cN)\ge \frac{1}{\sqrt{1-\epsilon}}$. We show that $\{M,N,K\}$ is a feasible solution to the dual problem, where $M=N=\frac{I}{2}$ and $K=-\frac{1}{2\sqrt{1-\epsilon}}O$.
To be specific, we have
\begin{align}
    &\tr_{A_iA_o}[(K^T_{A_i}\otimes J^T_{\cN_{A_oA'_i}}\otimes O_{A'_o})(J_{\id_{A_iA_o}}\otimes I_{A'_iA'_o})]\nonumber\\
    &= \tr_{A_iA_o}[(-\frac{1}{2\sqrt{1-\epsilon}}O\otimes\sum_{i,j}\ket{j}\bra{i}\otimes\cN^T(\ket{i}\bra{j})\otimes O)(\sum_{m,n}\ket{m}\bra{n}\otimes\ket{m}\bra{n}\otimes I\otimes I)]\\
    &= -\frac{1}{2\sqrt{1-\epsilon}} \tr_{A_iA_o}[\sum_{i,j,m,n}O\ket{m}\bra{n}\otimes\ket{j}\bra{i}m\rangle\bra{n}\otimes\cN^T(\ket{i}\bra{j})\otimes O]\\
    &= -\frac{1}{2\sqrt{1-\epsilon}} \sum_{i,j,n}\tr_{A_iA_o}[O\ket{i}\bra{n}\otimes\ket{j}\bra{n}\otimes\cN^T(\ket{i}\bra{j})\otimes O]\\
    &= -\frac{1}{2\sqrt{1-\epsilon}}\sum_{i,j,n}\tr[O\ket{i}\bra{n}]\tr[\ket{j}\bra{n}] \cN^T(\ket{i}\bra{j})\otimes O\\
    &= -\frac{1}{2\sqrt{1-\epsilon}} \sum_{i,j}O_{ji}\cN^T(\ket{i}\bra{j})\otimes O,
\end{align}
where $O_{ji} \equiv \bra{j}O\ket{i}$. Since $O\in\{X,Y\}$ only has non-zero elements on the anti-diagonal, we have
\begin{align}
    &\tr_{A_iA_o}[(K^T_{A_i}\otimes J^T_{\cN_{A_oA'_i}}\otimes O_{A'_o})(J_{\id_{A_iA_o}}\otimes I_{A'_iA'_o})]\nonumber\\
    &= -\frac{1}{2\sqrt{1-\epsilon}}(O_{10}\cN^T(\ket{0}\bra{1})+O_{01}\cN^T(\ket{1}\bra{0}))\otimes O\\
    &= -\frac{\sqrt{1-\epsilon}}{2\sqrt{1-\epsilon}}(O_{10}\ketbra{1}{0}+O_{01}\ketbra{0}{1})\otimes O\\
    &= -\frac{1}{2}O\otimes O,
\end{align}
where the second equality follows from Eq.~\eqref{seq:gad_act_on_state}. This means that
\begin{align}
    M_{A'_i}\ox I_{A'_o} + \tr_{A_iA_o}[(K_{A_i}^T\ox J_{\cN_{A_oA'_i}}^T\ox O_{A'_o}) (J_{\id_{A_iA_o}}\ox I_{A'_iA'_o})] &= \frac{1}{2}I\otimes I-\frac{1}{2}O\otimes O\geq 0,\\
    N_{A'_i}\ox I_{A'_o} - \tr_{A_iA_o}[(K_{A_i}^T\ox J_{\cN_{A_oA'_i}}^T\ox O_{A'_o}) (J_{\id_{A_iA_o}}\ox I_{A'_iA'_o})] &= \frac{1}{2}I\otimes I+\frac{1}{2}O\otimes O\geq 0.
\end{align}
Thus, $\{M,N,K\}$ is a feasible solution to the dual SDP~\eqref{eq:dual_sdp}, which means that $\gamma_O(\cN) \ge -\tr[KO] =\frac{1}{\sqrt{1-\epsilon}}$. Combining this with the primal part, we conclude that $\gamma_O(\cN) =\frac{1}{\sqrt{1-\epsilon}}$ for $O\in\{X,Y\}$ and $\cN$ being a single-qubit GAD channel.
\end{proof}

From the proof, we also know that the above retriever $\cD$ is optimal.
Moreover, since the Choi matrices of $\cD_1$ and $\cD_2$ are given already, it is trivial to derive the corresponding Kraus operators, which are $E_{\cD_1}=\left\{\sqrt{\frac{1}{2}}E_i \middle| E_i\in\{I,X,Y,Z\}, E_iO=OE_i\right\}$ and  $E_{\cD_2}=\left\{\sqrt{\frac{1}{2}}E_i \middle| E_i\in\{I,X,Y,Z\}, E_iO=-OE_i\right\}$.

It is interesting to note that there is a connection between the retrieving cost and the spectral properties of the noisy channel and the observable of interest. Let
$O = \begin{pmatrix}
O_{00} & O_{01}\\
O_{10} & O_{11}
\end{pmatrix}$
be a single-qubit observable. Then,
\begin{align}
    \cN_{\rm GAD}^\dagger(O) = \sum_{i=0}^{3} E_i^\dagger O E_i =
    \begin{pmatrix}
        O_{00} - \epsilon(1-p)(O_{00}-O_{11}) & \sqrt{1-\epsilon}O_{01}\\
        \sqrt{1-\epsilon}O_{10} & O_{11} + \epsilon p(O_{00}-O_{11})
    \end{pmatrix}.
\end{align}
For $O\in{X,Y}$, we have $\cN_{\rm GAD}^\dagger(O) = \sqrt{1-\epsilon}O$ and hence $\tr[\cN_{\rm GAD}(\rho) O] = \tr[\rho \cN_{\rm GAD}^\dagger(O)] = \sqrt{1-\epsilon} \tr[\rho O].$ Thus, we need a retriever to scales the expectation value back, and such a retriever corresponds to a retrieving cost of $1/\sqrt{1-\epsilon}$, which aligns with the above proposition. The same phenomena are observed for the mixed Pauli noises with Pauli observables (see Proposition~\ref{props_mixed_pauli}).

\begin{proposition}
Given an observable $O=Z$ and a GAD channel $\cN$, the minimum cost $\gamma_O(\cN)$ to retrieve the information $\tr[\rho O]$ is $\gamma_O(\cN) = \frac{|1-2p|\epsilon+1}{1-\epsilon}$, and the corresponding retriever $\cD$ can be written in the form of a Choi matrix
$$
J_{\cD} = c_1 J_{\cD_1} + c_2 J_{\cD_1},
$$
where $c_1=-c_2=\frac{|1-2p|\epsilon+1}{2(1-\epsilon)}$, and \begin{align}
    J_{\cD_1} &= \frac{1}{2(|1-2p|\epsilon+1)}Z\ox Z + \frac{\epsilon(1-2p)\epsilon}{2(|1-2p|\epsilon+1)}I\ox Z+\frac{1}{2}I\ox I\\
    J_{\cD_2} &= -\frac{1}{2(|1-2p|\epsilon+1)}Z\ox Z - \frac{\epsilon(1-2p)}{2(|1-2p|\epsilon+1)}I\ox Z+\frac{1}{2}I\ox I.
\end{align}
\end{proposition}

\begin{proof}
First, we are going to prove $\gamma_O(\cN)\le\frac{|1-2p|\epsilon+1}{(1-\epsilon)}$ by showing that the retriever $\cD$ is a feasible solution to the SDP~\eqref{eq:primal SDP}.
To be specific, we have
\begin{align}
    \cD\circ\cN(\rho) &= \cD(\rho')\\
    &= \tr_A[(\rho'^{T}\ox I)\cdot J_{\cD}] \\
    &= \tr_A\left[(\rho'^{T}\ox I) \cdot \frac{|1-2p|\epsilon+1}{2(1-\epsilon)} \left(\frac{1}{|1-2p|\epsilon+1} Z\ox Z + \frac{\epsilon(1-2p)}{|1-2p|\epsilon+1} I\ox Z\right)\right]\\
    &= \frac{1}{2(1-\epsilon)}Z \tr[\rho' Z] + \frac{\epsilon(1-2p)}{2(1-\epsilon)} Z\tr[\rho'].
    \end{align}
Therefore, we further have
\begin{align}
    \tr[Z \cD \circ \cN(\rho)]&= \frac{1}{2(1-\epsilon)} \tr[Z^2] \tr[\rho'Z] + \frac{\epsilon(1-2p)}{2(1-\epsilon)}\tr[Z^2]\tr[\rho']\\
    &= \frac{1}{(1-\epsilon)}((1-\epsilon)\rho_{00}+p\epsilon(\rho_{00}+\rho_{11})-\rho_{11}-\epsilon\rho_{00} \nonumber\\
    &\qquad +p\epsilon(\rho_{00}+\rho_{11})+\epsilon(1-2p)(\rho_{00}+\rho_{11}))\\
    &= \rho_{00}-\rho_{11}\\
    &= \tr[Z\rho],
\end{align}
where the second equality follows from Eq.~\eqref{seq:gad_act_on_state} and the fact that $Z^2=I$.
Hence, the retriever $\cD$ is a feasible solution to retrieve information, implying that $\gamma_O(\cN) \le |c_1|+|c_2| = \frac{|1-2p|\epsilon+1}{1-\epsilon}$.

Second, we use the dual SDP~\eqref{eq:dual_sdp} to show the cost $\gamma_O(\cN)\ge\frac{|1-2p|\epsilon+1}{1-\epsilon}$. To prove this, we are going to split into several cases: $0\leq p<\frac{1}{2}$, $p=\frac{1}{2}$, and $\frac{1}{2}<p\leq 1$. First note that
\begin{align}
    &\tr_{A_iA_o}[(K^T_{A_i}\otimes J^T_{\cN_{A_oA'_i}}\otimes O_{A'_o})(J_{\id_{A_iA_o}}\otimes I_{A'_iA'_o})]\nonumber\\
    &= \tr_{A_iA_o}[(K^T \otimes \sum\limits_{i,j}\ket{j}\bra{i}\otimes\cN^T(\ket{i}\bra{j})\otimes Z)(\sum\limits_{m,n}\ket{m}\bra{n}\otimes\ket{m}\bra{n}\otimes I\otimes I)]\\
    &= \tr_{A_iA_o}[\sum\limits_{i,j,m,n}K^T\ket{m}\bra{n} \otimes \ketbra{j}{i}m\rangle\bra{n}\otimes\cN^T(\ket{i}\bra{j})\otimes Z]\\
    &=\sum\limits_{i,j,n}\tr[K^T\ket{i}\bra{n}]\tr[\ket{j}\bra{n}]\cN^T(\ket{i}\bra{j})\otimes Z\\
    &=\sum\limits_{i,j}\tr[K^T\ket{i}\bra{j}]\cN^T(\ket{i}\bra{j})\otimes Z.
\end{align}

\begin{itemize}
    \item When $0\leq p<\frac{1}{2}$, we set the dual variables $\{M,N,K\}$ as $M=N=\ket{0}\bra{0}, K=-\frac{1}{2}I+\frac{2p\epsilon-1-\epsilon}{2(1-\epsilon)}Z$ and prove that they form a feasible solution to the dual SDP.
    \begin{align}
    &\tr_{A_iA_o}[(K^T_{A_i}\otimes J^T_{\cN_{A_oA'_i}}\otimes O_{A'_o})(J_{\id_{A_iA_o}}\otimes I_{A'_iA'_o})]\nonumber\\
    &=\sum_{i,j} \left(-\frac{1}{2}\tr[\ket{i}\bra{j}] + \frac{2p\epsilon-1-\epsilon}{2(1-\epsilon)}\tr[Z\ket{i}\bra{j}]\right) \cN^T(\ket{i}\bra{j})\otimes Z\\
    &= \sum_{i} \left(-\frac{1}{2}+ (-1)^i \frac{2p\epsilon-1-\epsilon}{2(1-\epsilon)}\right) \cN^T(\ket{i}\bra{i})\otimes Z\\
    &=\left(\frac{p\epsilon-1}{1-\epsilon}\cN^T(\ket{0}\bra{0})+\frac{\epsilon-p\epsilon}{1-\epsilon}\cN^T(\ket{1}\bra{1})\right)\otimes Z\\
    &= -\ket{0}\bra{0}\otimes Z\\
    &= -\ket{00}\bra{00}+\ket{01}\bra{01}.
    \end{align}
    Therefore,
    \begin{align}
        M_{A'_i}\ox I_{A'_o} + \tr_{A_iA_o}[(K_{A_i}^T\ox J_{\cN_{A_oA'_i}}^T\ox O_{A'_o}) (J_{\id_{A_iA_o}}\ox I_{A'_iA'_o})] &= 2\ket{01}\bra{01}\geq 0\\
        N_{A'_i}\ox I_{A'_o} - \tr_{A_iA_o}[(K_{A_i}^T\ox J_{\cN_{A_oA'_i}}^T\ox O_{A'_o}) (J_{\id_{A_iA_o}}\ox I_{A'_iA'_o})] &= 2\ket{00}\bra{00}\geq 0.
    \end{align}

    \item When $p=\frac{1}{2}$, we set the dual variables $\{M,N,K\}$ as $M=N=\frac{1}{2}I, K=-\frac{1}{2(1-\epsilon)}Z$ and prove that they form a feasible solution to the dual SDP.
    \begin{align}
    &\tr_{A_iA_o}[(K^T_{A_i}\otimes J^T_{\cN_{A_oA'_i}}\otimes O_{A'_o})(J_{\id_{A_iA_o}}\otimes I_{A'_iA'_o})]\nonumber\\
    &= \sum_{i,j} \tr\left[-\frac{1}{2(1-\epsilon)}Z\ket{i}\bra{j}\right] \cN^T(\ket{i}\bra{j})\otimes Z\\
    &= -\frac{1}{2(1-\epsilon)}(\cN^T(\ket{0}\bra{0})-\cN^T(\ket{1}\bra{1})) \otimes Z\\
    &= -\frac{1}{2}Z\ox Z.
    \end{align}
    Therefore,
    \begin{align}
        M_{A'_i}\ox I_{A'_o} + \tr_{A_iA_o}[(K_{A_i}^T\ox J_{\cN_{A_oA'_i}}^T\ox O_{A'_o}) (J_{\id_{A_iA_o}}\ox I_{A'_iA'_o})] &= \frac{1}{2}I\ox I-\frac{1}{2}Z\ox Z\geq 0,\\
        N_{A'_i}\ox I_{A'_o} - \tr_{A_iA_o}[(K_{A_i}^T\ox J_{\cN_{A_oA'_i}}^T\ox O_{A'_o}) (J_{\id_{A_iA_o}}\ox I_{A'_iA'_o})] &= \frac{1}{2}I\ox I+\frac{1}{2}Z\ox Z\geq 0.
    \end{align}

    \item When $\frac{1}{2}<p\leq 1$, we set the dual variables $\{M,N,K\}$ as $M=N=\ket{1}\bra{1}, K=\frac{1}{2}I+\frac{\epsilon-2p\epsilon-1}{2(1-\epsilon)}$ and prove that they form a feasible solution to the dual SDP.
    \begin{align}
    &\tr_{A_iA_o}[(K^T_{A_i}\otimes J^T_{\cN_{A_oA'_i}}\otimes O_{A'_o})(J_{\id_{A_iA_o}}\otimes I_{A'_iA'_o})]\nonumber\\
    &=\sum_{i,j} \left(\frac{1}{2}\tr[\ket{i}\bra{j}] + \frac{\epsilon-2p\epsilon-1}{2(1-\epsilon)}\tr[Z\ket{i}\bra{j}]\right) \cN^T(\ket{i}\bra{j})\otimes Z\\
    &= \sum_{i}\left(\frac{1}{2}+ (-1)^i \frac{\epsilon-2p\epsilon-1}{2(1-\epsilon)}\right) \cN^T(\ket{i}\bra{i})\otimes Z\\
    &= \left(\frac{-p\epsilon}{1-\epsilon}\cN^T(\ket{0}\bra{0}) + \frac{p\epsilon+1-\epsilon}{1-\epsilon}\cN^T(\ket{1}\bra{1})\right) \otimes Z\\
    &= \ket{1}\bra{1}\otimes Z\\
    &= \ket{10}\bra{10}-\ket{11}\bra{11}.
    \end{align}
Therefore,
\begin{align}
    M_{A'_i}\ox I_{A'_o} + \tr_{A_iA_o}[(K_{A_i}^T\ox J_{\cN_{A_oA'_i}}^T\ox O_{A'_o}) (J_{\id_{A_iA_o}}\ox I_{A'_iA'_o})] &= 2\ket{10}\bra{10}\geq 0,\\
    N_{A'_i}\ox I_{A'_o} - \tr_{A_iA_o}[(K_{A_i}^T\ox J_{\cN_{A_oA'_i}}^T\ox O_{A'_o}) (J_{\id_{A_iA_o}}\ox I_{A'_iA'_o})] &= 2\ket{11}\bra{11}\geq 0.
\end{align}
\end{itemize}

Hence, we conclude that under different $p$ values, the corresponding solutions $\{M,N,K\}$ are all feasible. It can be easily verified that $-\tr[KO] = \frac{|1-2p|\epsilon+1}{1-\epsilon}$ for all these cases, which implies that the inequality $\gamma_O(\cN) \ge \frac{|1-2p|\epsilon+1}{1-\epsilon}$ always holds. Combining it with the primal part, we have $\gamma_O(\cN) = \frac{|1-2p|\epsilon+1}{1-\epsilon}$.
\end{proof}

From the proof, we also know that the above retriever $\cD$ is optimal.
Moreover, since the Choi matrix has been given, the corresponding Kraus operators of the retriever can be easily derived. 
\begin{itemize}
    \item When $0\leq p < \frac{1}{2}$, 
    the Kraus operators for $\cD_1$ and $\cD_2$ are $E_{\cD_1}=\{\alpha\ket{1}\bra{1},\ket{0}\bra{0},\beta\ket{0}\bra{1}\}$ and 
    $E_{\cD_2}=\{\ket{1}\bra{0},\alpha\ket{0}\bra{1},\beta\ket{0}\bra{1}\}$, where $\alpha=\sqrt{\frac{1}{1+\epsilon|1-2p|}}$ and $\beta=\sqrt{\frac{\epsilon|1-2p|}{1+\epsilon|1-2p|}}$.
    
    \item When $\frac{1}{2}\leq p \leq 1$, the Kraus operators for $\cD_1$ and $\cD_2$ are $E_{\cD_1}=\{\ket{1}\bra{1},\alpha\ket{0}\bra{0},\beta\ket{1}\bra{0}\}$ and $E_{\cD_2}=\{\ket{0}\bra{1},\alpha\ket{1}\bra{1},\beta\ket{1}\bra{1}\}$, where $\alpha=\sqrt{\frac{1}{1+\epsilon|1-2p|}}$ and $\beta=\sqrt{\frac{\epsilon|1-2p|}{1+\epsilon|1-2p|}}$.
\end{itemize}

\section{Retrieving Cost for Mixed Pauli Channel}
The mixed Pauli channel is a common noise model in quantum computers. For single-qubit cases, a quantum state corrupted by mixed Pauli becomes 
\begin{equation}\label{eq:1_qubit_depo}
    \cN_{\rm Pauli}(\rho) = p_i\rho + p_x X\rho X + p_y Y\rho Y + p_z Z\rho Z,
\end{equation}
where $p_i,p_x,p_y,p_z$ are the corresponding probabilities with $p_i+p_x+p_y+p_z=1$ and $0 \le p \le 1$ for each $p \in \{p_i,p_x,p_y,p_z\}$. For general $n$-qubit cases, the noisy state is 
\begin{equation}\label{eq:n_qubit_depo}
    \cN_{\rm Pauli}(\rho)=\sum_\sigma p_\sigma \sigma \rho \sigma,
\end{equation}
where the sum is over all the $n$-qubit Pauli operators $\{\sigma\}$, and $\{p_\sigma\}$ are the corresponding probabilities with $\sum_\sigma p_\sigma=1$ and $0\le p_\sigma\le 1$. Note that the depolarizing channel is a special case of the mixed Pauli channels, where $p_x=p_y=p_z$.

\begin{proposition}\label{props_mixed_pauli}
Given an $n$-qubit observable $O=\bigotimes_{i=1}^n \sigma_i$, where $\sigma_i\in\{X,Y,Z,I\}$, and a mixed Pauli channel $\cN$. The minimum cost $\gamma_O(\cN)$ to retrieve the information $\tr[\rho O]$ is $\gamma_O(\cN) = \frac{1}{\sum_{\sigma^+} p_{\sigma^+}-\sum_{\sigma^-} p_{\sigma^-}}$, and the corresponding retriever $\cD$ can be written in the form of a Choi matrix $J_{\cD} = c_1 J_{\cD_1} + c_2 J_{\cD_1}$, where $J_{\cD_1} = \frac{1}{2^n}I^{\ox 2n}+\frac{1}{2^n}(O^T \ox O), J_{\cD_2} = \frac{1}{2^n}I^{\ox 2n}-\frac{1}{2^n}(O^T \ox O)$ and $c_1=-c_2=\frac{1}{2(\sum_{\sigma^+} p_{\sigma^+}-\sum_{\sigma^-} p_{\sigma^-})}$. The $\{\sigma^+\}$ are all the Pauli operators that commute with the observable, i.e., $\sigma^+\in P_n, [\sigma^+, O]=0$, and correspondingly $\sigma^-\in P_n, \{\sigma^-, O\}=0$.
\end{proposition}

\begin{proof}
First, we are going to prove $\gamma_O(\cN) \le \frac{1}{\sum_{\sigma^+} p_{\sigma^+}-\sum_{\sigma^-} p_{\sigma^-}}$. We prove this by utilizing SDP~\eqref{eq:primal SDP} and  showing that the retriever $\cD$ is a feasible solution.

For an arbitrary state $\rho$, after the mixed Pauli channel, it becomes $\rho' = \cN(\rho)=\sum_{\sigma^+}p_{\sigma^+}\sigma^+ \rho \sigma^+ + \sum_{\sigma^-}p_{\sigma^-}\sigma^- \rho \sigma^-$. Then, we have
\begin{align}
    \cD\circ\cN(\rho) &= \cD(\rho')\\
    &= \tr_A[(\rho'^{T}\ox I)\cdot J_{\cD}^O] \\
    &= \frac{1}{2^n \cdot (\sum_{\sigma^+} p_{\sigma^+}-\sum_{\sigma^-} p_{\sigma^-})}\tr_A[(\rho'^{T}\ox I) \cdot (O^T\ox O)]\\
    &= \frac{1}{2^n \cdot (\sum_{\sigma^+} p_{\sigma^+}-\sum_{\sigma^-} p_{\sigma^-})} \tr_A[\rho'^T O^T\ox IO]\\
    &= \frac{1}{2^n \cdot (\sum_{\sigma^+} p_{\sigma^+}-\sum_{\sigma^-} p_{\sigma^-})} O  \tr[\rho'^T O^T]\\
\end{align}
Since the transpose is linear, we have $\tr[\rho'^T O^T] = \tr[(\rho' O)^T] = \tr[\rho' O]$.
Therefore,
\begin{align}
    \cD\circ\cN(\rho) &= \frac{1}{2^n \cdot (\sum_{\sigma^+} p_{\sigma^+}-\sum_{\sigma^-} p_{\sigma^-})} O  \tr[\rho' O]\\
    &= \frac{1}{2^n \cdot (\sum_{\sigma^+} p_{\sigma^+}-\sum_{\sigma^-} p_{\sigma^-})} O \tr\left[\left(\sum_{\sigma^+}p_{\sigma^+}\sigma^+ \rho \sigma^+ + \sum_{\sigma^-}p_{\sigma^-}\sigma^- \rho \sigma^-\right) O\right] \\
    &= \frac{1}{2^n \cdot (\sum_{\sigma^+} p_{\sigma^+}-\sum_{\sigma^-} p_{\sigma^-})} O \tr\left[\left(\sum_{\sigma^+}p_{\sigma^+}\sigma^+ \rho O \sigma^+ - \sum_{\sigma^-}p_{\sigma^-}\sigma^- \rho O \sigma^-\right)\right]\\
    &= \frac{1}{2^n \cdot (\sum_{\sigma^+} p_{\sigma^+}-\sum_{\sigma^-} p_{\sigma^-})} O \tr\left[\left(\sum_{\sigma^+}p_{\sigma^+} \rho O \sigma^+ \sigma^+ - \sum_{\sigma^-}p_{\sigma^-} \rho O \sigma^- \sigma^-\right)\right] \\
    &= \frac{1}{2^n \cdot (\sum_{\sigma^+} p_{\sigma^+}-\sum_{\sigma^-} p_{\sigma^-})} O \tr\left[\left(\sum_{\sigma^+}p_{\sigma^+} - \sum_{\sigma^-}p_{\sigma^-}\right) \rho O\right]\\
    &= \frac{1}{2^n} O \tr[\rho O],
\end{align}
where the third equality follows from the commutative properties of $\{\sigma_+\}$ and $\{\sigma_-\}$ with respect to the observable, the fourth equality follows from the cyclic property of trace, and the fifth equality follows from the identities that $\sigma^+\sigma^+ = \sigma^-\sigma^- = I^{\ox n}$.
Then, it is easy to check that
\begin{align}
    \tr[O \cD\circ\cN(\rho)] &= \frac{1}{2^n} \tr[O^2]\tr[\rho O]\\
    &= \frac{1}{2^n}\tr[I^{\ox n}]\tr[\rho O]\\
    &= \tr[\rho O],
\end{align}
which means that the retriever $\cD$ is a feasible protocol, and thus we have $\gamma_O(\cN) \le |c_1|+|c_2| = \frac{1}{\sum_{\sigma^+} p_{\sigma^+}-\sum_{\sigma^-} p_{\sigma^-}}$.

Second, we use the dual SDP~\eqref{eq:dual_sdp} to show that $\gamma_O(\cN) \ge \frac{1}{\sum_{\sigma^+}p_{\sigma^+}-\sum_{\sigma^-}p_{\sigma^-}}$. We set the dual variables as $M=N=\frac{1}{2^n}I, K = q O$, where $n$ is the number of qubits and $q \equiv -\frac{1}{2^n(\sum_{\sigma^+}p_{\sigma^+}-\sum_{\sigma^-}p_{\sigma^-})}$. We will show the variables $\{M,N,K\}$ is a feasible solution to the dual problem. 
\begin{align}
    &\tr_{A_iA_o}[(K^T_{A_i}\ox J^T_{\cN_{A_oA'_i}}\ox O_{A'_o})(J_{\id_{A_iA_o}}\ox I_{A'_iA'_o})]\nonumber\\
    &= \tr_{A_iA_o}\left[\left(qO^T \ox \sum_{i,j}\ket{j}\bra{i}\ox\sum_{\sigma}p_\sigma \sigma^T\ket{j}\bra{i}\sigma^T \ox O\right)\left(\sum_{m,n}\ket{m}\bra{n}\ox\ket{m}\bra{n}\ox I\ox I\right)\right]\\
    &= \tr_{A_iA_o}\left[\sum_{i,j,m,n}\sum_{\sigma}q O^T\ket{m}\bra{n}\ox \ket{j}\bra{i}m\rangle\bra{n}\ox p_\sigma \sigma^T\ket{j}\bra{i}\sigma^T\ox O\right]\\
    &= \sum_{i,j,n}\sum_{\sigma} q p_\sigma\tr[O^T\ket{i}\bra{n}\ox\ket{j}\bra{n}]\sigma^T\ket{j}\bra{i}\sigma^T\ox O\\
    &= \sum_{i,j,n}\sum_{\sigma} q p_\sigma\tr[O^T\ket{i}\bra{n}]\tr[\ket{j}\bra{n}]\sigma^T\ket{j}\bra{i}\sigma^T\ox O\\
    &= \sum_{i,j}\sum_{\sigma} q p_\sigma O^T_{ji}\sigma^T\ket{j}\bra{i}\sigma^T\ox O\\
    &= \sum_{\sigma} q p_\sigma \sigma^T O^T \sigma^T\ox O\\
    &= q \left(\sum_{\sigma^+} p_{\sigma^+} O^T\ox O-\sum_{\sigma^-} p_{\sigma^-} O^T\ox O\right)\\
    &= -\frac{1}{2^n}O^T\ox O,
\end{align}
where $O^T_{ji} \equiv \bra{j}O^T\ket{i}$.
This implies that
\begin{align}
    M_{A'_i}\ox I_{A'_o} + \tr_{A_iA_o}[(K_{A_i}^T\ox J_{\cN_{A_oA'_i}}^T\ox O_{A'_o}) (J_{\id_{A_iA_o}}\ox I_{A'_iA'_o})] &= \frac{1}{2^{2n}}I\ox I - \frac{1}{2^n}O^T\ox O\geq 0,\\
    N_{A'_i}\ox I_{A'_o} - \tr_{A_iA_o}[(K_{A_i}^T\ox J_{\cN_{A_oA'_i}}^T\ox O_{A'_o}) (J_{\id_{A_iA_o}}\ox I_{A'_iA'_o})] &= \frac{1}{2^{2n}}I\ox I + \frac{1}{2^n}O^T\ox O\geq 0,
\end{align}
which means that $\{M,N,K\}$ is a feasible solution to the dual SDP~\eqref{eq:dual_sdp}. Therefore, we have $\gamma_O(\cN) \ge -\tr[KO] = \frac{1}{\sum_{\sigma^+} p_{\sigma^+}-\sum_{\sigma^-} p_{\sigma^-}}$. Combining this with the primal part, we conclude that $\gamma_O(\cN) = \frac{1}{\sum_{\sigma^+} p_{\sigma^+}-\sum_{\sigma^-} p_{\sigma^-}}$.
\end{proof}

From the proof, we also know that the above retriever $\cD$ is optimal.
Moreover, from the given Choi matrix, one could derive that the corresponding Kraus operators of $\cD_1$ and $\cD_2$ are  
$$E_{\cD_1} = \left\{\sqrt{\frac{1}{2^{2n-1}}} E_i \middle| E_i\in\{I,X,Y,Z\}^{\ox n}, E_iO=OE_i \right\}$$ and $$E_{\cD_2} = \left\{\sqrt{\frac{1}{2^{2n-1}}} E_i \middle| E_i\in\{I,X,Y,Z\}^{\ox n}, E_iO=-OE_i \right\}.$$

Note that the depolarizing channel is a special case of the mixed Pauli channels, where all noise coefficients are the same. For the single-qubit case, the depolarizing channel is $\cN_{\rm depo}(\rho) = (1-\epsilon)\rho + \epsilon \frac{I}{2}$, which is equivalent to Eq.~\eqref{eq:1_qubit_depo} by setting $p_x=p_y=p_z=\frac{\epsilon}{4}$ and $p_i=1-p_x-p_y-p_z=1-\frac{3\epsilon}{4}$. The $n$-qubit depolarizing channel is $\cN_{n-{\rm depo}}(\rho) = (1-\epsilon)\rho + \epsilon \frac{I}{2^n}$, by setting the parameters in Eq.~\eqref{eq:n_qubit_depo} as $p_\sigma=\frac{\epsilon}{4^n}$ and $p_I=1-\sum_\sigma p_\sigma= 1-\epsilon\frac{4^n-1}{4^n}$, where $\{\sigma\}$ are the $n$-qubit Pauli operators excluding $I^{\ox n}$, and $\{p_\sigma\}$ are the corresponding probabilities. It is obvious that $I^{\ox n}$ commutes with any observable, i.e., $I\in \sigma^+$. Since the rest probabilities are identical, we have $\sum_{\sigma^+} p_{\sigma^+} - \sum_{\sigma^-} p_{\sigma^-} = 1-\epsilon$, which means that the retrieving cost for depolarizing channel is $\gamma_O(\cN_{n-{\rm depo}})=\frac{1}{1-\epsilon}$ for $O$ being an $n$-qubit Pauli operator.

\end{document}